\documentclass[12pt, reqno]{article}
\usepackage{amsmath}
\usepackage{amsfonts}
\usepackage{amsthm}
\usepackage[all]{xy}
\usepackage{hyperref}
\usepackage{longtable}
\usepackage{slashbox}
\usepackage[justification=centering]{caption}
\usepackage{multirow}
\usepackage{appendix}

\newtheorem{thm}{Theorem}[section]

\newcommand{\Rpq}{\mathbb{R}^{p,q}}
\newcommand{\Cpq}{Cl_{p,q}}

\newcommand{\R}{\mathbb{R}}
\newcommand{\C}{\mathbb{C}}
\newcommand{\om}{\omega}
\newcommand{\eps}{\epsilon}
\numberwithin{equation}{section}

\title{Bilinear Forms and Fierz Identities for Real Spin Representations}
\author{Eric O. Korman\footnote{\texttt{eric.korman@gmail.com}}, George Sparling \\ \\
Laboratory of Axiomatics \\
Department of Mathematics \\
University of Pittsburgh
}
\date{}

\begin{document}

\maketitle

\begin{abstract}
\noindent Given a real representation of the Clifford algebra corresponding to $\R^{p+q}$ with metric of signature $(p,q)$, we demonstrate the existence of two natural bilinear forms on the space of spinors.  With the Clifford action of $k$-forms on spinors, the bilinear forms allow us to relate two spinors with elements of the exterior algebra.  From manipulations of a rank four spinorial tensor introduced in $\cite{penroserindler}$, we are able to find a general class of identities which, upon specializing from four spinors to two spinors and one spinor in signatures (1,3) and (10,1), yield some well-known Fierz identities.  We will see, surprisingly, that the identities we construct are partly encoded in certain involutory real matrices that resemble the Krawtchouk matrices \cite{fein1}\cite{fein2}.
\end{abstract}

\section{Introduction}
The Fierz identities are relations among elements of the Clifford algebra, spinors, and the exterior algebra associated with the vector space $\R^n$ with metric $g$.  They have been used in Dirac's treatment of electron spin in signature (1,3) \cite{lounesto} as well as in M-theory (which uses signature (10,1)) \cite{mtheory}, \cite{11dsup}.  Like those in \cite{lounesto}, our Fierz identities are relations among spinors as opposed to spinor one-forms, which are the objects of interest in \cite{mtheory} and \cite{11dsup}.  However, our identities are more general than those in \cite{lounesto}: we have four-spinor identities instead of just one-spinor identities and our constructions work in arbitrary dimension and signature.  Our approach is an algebraic one, with the multiplicative group structure of the generators of the Clifford algebra playing an essential role in the derivations.  In the process of deriving the Fierz identities, we show the existence of a class of involutory real matrices, of dimension $n + 1 \times n + 1$ or $\frac{n+1}{2} \times \frac{n+1}{2}$.  We offer two proofs in the appendices that these matrices square to the identity.  One proof uses spinors and the other is a direct proof that uses contour integration.  \\

\noindent Although we restrict attention to real representations, we will see that a quaternionic structure arises when the Clifford algebra is isomorphic to a full matrix algebra over $\mathbb{H}$ (the quaternions).  When $\Cpq$ is isomorphic to a full matrix algebra over $\C$, we get a complex structure as well as an additional operator, which anti-commutes with $i$.  In some cases, this additional structure squares to the identity, in which case we can think of it as complex conjugation. \\

\noindent In the remainder of this section, we briefly summarize results from the general theory of Clifford algebras.  In the next section we define two natural bilinear forms on the space of spinors and then examine how we can use these to relate spinors to elements of the exterior algebra.  Finally, in the third section, we derive the various identities, considering the three cases where a Clifford algebra is isomorphic to a full matrix algebra over $\R, \C,$ or $\mathbb{H}$.  We also consider two cases of special interest to physics: the Clifford algebras associated with signature $(1,3)$ and $(10,1)$.  We show that our identities reduce to some familiar Fierz identities upon specialization from four spinors to one spinor.

\subsection{Clifford Algebras}
The \textit{Clifford algebra} of a real vector space $V$ with metric $g$ is the free algebra generated by $V$, modulo the relation
\begin{equation}
v^2 = g(v,v). \label{genrel}
\end{equation}
Replacing $v$ with $v + w$ in the above and expanding yields the relation
\begin{equation}
vw + wv = 2g(v,w). \label{vw}
\end{equation}
If $V=\R^{p+q}$ and $g$ has signature $(p,q)$ (i.e. $(\underbrace{+ + \ldots +}_{p-times} \underbrace{- - \ldots -}_{q-times}))$, then we denote the corresponding Clifford algebra by $\Cpq$.  We denote the image of the natural inclusion map $\R^{p+q} \hookrightarrow \Cpq$ by $\Rpq$. \\

\noindent If $\{e_i : 1 \le i \le p+q \}$ is the standard basis for $\R^{p+q}$ then from $(\ref{genrel})$ and $(\ref{vw})$ we have that
\begin{equation*}
e_i^2 = \begin{cases}
1 &\text{ if $1 \le i \le p$} \\
-1 &\text{ if $p + 1 \le i \le p + q$}
\end{cases}
\end{equation*}
and
\begin{equation*}
e_i e_j = -e_j e_i ~~ (i \ne j).
\end{equation*} \\

\noindent Clearly $\{ e_1^{i_1} e_2^{i_2} \ldots e_{p+q}^{i_{p+q}} : i_k = 0 \text{ or } 1 \}$ spans $\Cpq$, so that $\dim \Cpq \le 2^{p+q}$.  It can be shown that if $p - q \ne 1$ (mod 4) then any algebra generated by a set $\{ e_1 \ldots e_{p+q}\}$ satisfying the above relations must have dimension $2^{p+q}$.  If $p - q = 1$ (mod 4) then it is possible for the dimension to be $2^{p+q-1}$, with $e_1 e_2 \ldots e_{p+q} = \pm 1$ \cite{porteous}.  The element $e_1 e_2 \ldots e_{p+q}$ is canonical \cite{lawson}, is denoted by $\gamma$, and called the $\textit{pseudoscalar}$.  A straightforward computation shows that
\begin{equation}
\gamma^2 = (-1)^{\frac{(p+q)^2+q-p}{2}} \label{g2}
\end{equation}
and
\begin{equation}
\gamma u = \begin{cases}
u \gamma &\text{iff $p+q$ is odd or $u$ is even} \\
-u \gamma &\text{iff $p+q$ is even and $u$ is odd.}
\end{cases} \label{gcom}
\end{equation} \\

\noindent If $I = \{ i_1, i_2, \ldots i_n \}$, where each $i_j \in \mathbb{N}$ and $i_j < i_{j+1}$, then we use the notation $e_I$ for $\prod_{j=1}^n e_{i_j}$ and define $e_\emptyset = 1$.  We denote the grade involution by $\alpha$, where $\alpha(e_{i_1} e_{i_2} \ldots e_{i_n}) = (-1)^n e_{i_1} e_{i_2} \ldots e_{i_n}$.  We also make use of an algebra anti-involution $\tilde{ }$ called reversion, with $\widetilde{e_{i_1} \ldots e_{i_k}} = e_{i_k} \ldots e_{i_1}$. \\

\noindent We define the $Pin$ and $Spin$ subgroups of $\Cpq$ by
\begin{align*}
Pin(p,q) &= \{ v_1 v_2 \ldots v_n : v_i \in \Rpq, g(v_i,v_i)  = \pm 1 \} \\
Spin(p,q) &= \{ u \in Pin(p,q) : \alpha(u) = u \}.
\end{align*}
We give an action of $Pin(p,q)$ on $\Rpq$ by 
\begin{equation*}
u(v) = uv\alpha(u^{-1}), \text{ $u \in Pin(p,q)$, $v \in \Rpq$}.
\end{equation*} 
This action gives a 2-to-1 homomorphism from $Pin(p,q)$ to $O(p,q)$, the group of orthogonal transformations of signature $(p,q)$.  When restricted to $Spin(p,q)$ we get a 2-to-1 homomorphism to $SO(p,q)$, the special orthogonal group in signature $(p,q)$ \cite{porteous}.

\subsection{Representations of Clifford Algebras}
We can always represent $\Cpq$ as the set of all $n \times n$ matrices with entries in $\R, \C$, or $\mathbb{H}$.  We denote the set of all $n \times n$ matrices with entries in $\mathbb{F}$ by $\mathbb{F}[n]$.  The representation space is called the $\textit{space of spinors}$.  If $p-q \ne 1$ (mod 4), then the representation is unique.  Otherwise, there are two inequivalent representations; one has $\gamma = 1$ and the other has $\gamma = -1$, where $1$ is the identity matrix.  Furthermore we can always chose our representations such that $e_I^\dagger = e_I$ if and only if $e_I^2 = 1$ and $e_I^\dagger = -e_I$ if and only if $e_i^2 = -1$, where $\dagger$ is the conjugate transpose.  Therefore 
\begin{equation}
e_I^\dagger = e_I^{-1}. \label{tis-1}
\end{equation}

\noindent The full matrix algebra that $\Cpq$ is isomorphic to is determined by the quantity $\tau = q - p - 1$ mod 8:
\begin{equation}
\Cpq \simeq \begin{cases}
\R[2^{[p+q/2]}] &\text{ if $\tau = $5,6, or 7} \\
\mathbb{H}[2^{[p+q/2] - 1}] &\text{ if $\tau = $1,2, or 3} \\
\C[2^{p+q-1/2}] &\text{ if $\tau = $0, or 4}
\end{cases}, \label{reps}
\end{equation}
where $[r]$ denotes the integer part of $r$.  We say that $\Cpq$ is of type $\R, \mathbb{H}$, or $\C$ accordingly. \\

\subsubsection{Corner and Subordinate Algebras}
As mentioned above, if $p - q \equiv 1$ (mod 4) then there exist two inequivalent representations for $\Cpq$ on the same full matrix algebra.  From (\ref{reps}) we see that this happens if and only if $p+q$ is odd and $\Cpq$ is of type $\R$ or $\mathbb{H}$.  We call these Clifford algebras \textit{corner algebras} since we can find a real representation for any Clifford algebra from a representation of its closest corner algebra.  We therefore use the name \textit{subordinate algebra} for type $\R$ and $\mathbb{H}$ Clifford algebras that are not corner algebras. \\

\noindent Looking at (\ref{reps}) we see that type $\R$ algebras and type $\mathbb{H}$ algebras occur with the same frequency.  They also occur for both even and odd values of $p+q$ and both give rise to corner algebras.  Conversely, type $\C$ algebras occur less frequently and only in odd dimension.  Throughout this paper we will see that type $\R$ and $\mathbb{H}$ algebras can be treated similarly.  For example, in the next section we show that if we use a real representation for a type $\mathbb{H}$ algebra, quaternionic structure emerges in the form of three operators which satisfy the quaternionic relations and, when added to the representation of the Clifford algebra, generate the entire matrix algebra.  Furthermore, when we construct the cross-symmetry matrices we find that those for the type $\R$ and $\mathbb{H}$ algebras are identical except for a factor of two.   Conversely, real representations of type $\C$ algebras are not as neat, as an extra operator emerges that does not always have a clear interpretation (sometimes it can be thought of as conjugation).  For these reasons, we will only focus on Fierz identities for type $\R$ and $\mathbb{H}$ algebras.

\subsubsection{Real Representations of type $\mathbb{H}$ Algebras} \label{realreptypeH}
Suppose that $\Cpq \simeq \mathbb{H}^n$.  We can get a representation on $\R^{4n}$ by making the replacements
\begin{gather*}
i \to \left(\begin{array}{cccc}0 & -1 & 0 & 0 \\1 & 0 & 0 & 0 \\0 & 0 & 0 & -1 \\0 & 0 & 1 & 0\end{array}\right) \\
j \to \left(\begin{array}{cccc}0 & 0 & -1 & 0 \\0 & 0 & 0 & 1 \\1 & 0 & 0 & 0 \\0 & -1 & 0 & 0\end{array}\right) \\
k \to \left(\begin{array}{cccc}0 & 0 & 0 & -1 \\0 & 0 & -1 & 0 \\0 & 1 & 0 & 0 \\1 & 0 & 0 & 0\end{array}\right) \\
1 \to \left(\begin{array}{cccc}1 & 0 & 0 & 0 \\0 & 1 & 0 & 0 \\0 & 0 & 1 & 0 \\0 & 0 & 0 & 1\end{array}\right). 
\end{gather*}
Since $\dim_\R \mathbb{H}[n] = 4 n^2$ and $\dim_\R \R[4n] = 16 n^2$, the representation of $\Cpq$ on $\R^{4n}$ is not surjective (note however that the (real) dimension of the representation space is the same-- $4n$).  However, we will now show that there are always three matrices, which we suggestively call $I, J$, and $K$, we can add to get all of $\R[4n]$, i.e. $\Cpq \cup \{I, J, K\}$ generates $\R[4n]$.  These matrices satisfy the quaternionic relations so we can give an action of $\mathbb{H}$ on the space of spinors, $\mathcal{S}$, by
\begin{equation*}
(q_0 + q_1 i + q_2 j + q_3 k) \psi = (q_0 + q_1 I + q_2 J + q_3 K) \psi, ~~q_i \in \R, \psi \in \mathcal{S}.
\end{equation*}
It turns out that these matrices commute with $\Cpq$. \\

\noindent Since $\Cpq$ is of type $\mathbb{H}$ we have that $q - p - 1 \equiv$ 1, 2, or 3 (mod 8) (\ref{reps}).  Assume first that $q - p - 1 \equiv 1$.  Then $q - (p+2) - 1 \equiv 7$ so that $Cl_{p+2, q}$ is type $\R$.  We then get a real representation of $\Cpq$ by taking $\{ e_3, e_4, \ldots, e_{p+q+2} \}$ as generators, where $\{ e_1, e_2, \ldots, e_{p+q+2} \}$ are the standard generators for $Cl_{p+2,q}$.  Let $I = e_1 e_2, J = \gamma e_2, K = \gamma e_1$, where $\gamma = \prod_{i=1}^{p+q+2} e_i$  is the $Cl_{p+2,q}$ pseudoscalar.  Since $p+q+2$ is even, $\gamma$ anti-commutes with vectors so that $I, J$, and $K$ anti-commute with each other but commute with everything in $\Cpq$.  Furthermore, since $q - p - 1 \equiv 1$ (mod 8), $q - (p + 2) \equiv 0$ (mod 4) so that $\gamma^2 = 1$ by (\ref{g2}).  Thus $I^2 = J^2 = K^2 = -1$.  \\

\noindent The generators of $\Cpq$ along with the $I, J$, and $K$ operators now generate the entire real endomorphism algebra of $\mathcal{S}$ since
\begin{align*}
e_1 = \pm e_3 e_4 \ldots e_{p+q+2} J \\
e_2 = \pm e_3 e_4 \ldots e_{p+q+2} K
\end{align*}
and $Cl_{p+2,q} \simeq End_\R(S)$. \\

\noindent Consider now the case $q - p - 1 \equiv 3$ (mod 8).  Then $(q+2) - p - 1 \equiv 5$ (mod 8) so that $Cl_{p,q+2}$ is type $\R$.  As before, if $\{e_1, e_2, \ldots, e_{p+q+2} \}$ generates $Cl_{p,q+2}$, we  use $\{ e_1, \ldots, e_{p+q} \}$ to generate $\Cpq$ and give quaternionic structure with $I = e_{p+q+1} e_{p+q+2}, J = \gamma e_{p+q+2}, K = \gamma e_{p+q+1}$  (note that now $(q + 2) - p \equiv 6$ (mod 8) so that $\gamma^2 = -1$).  As in the previous case, $\{e_1, \ldots, e_{p+q} \} \cup \{ I, J, K \}$ generate $End_\R(\mathcal{S})$ and $I, J$, and $K$ are in the center of $\Cpq$. \\

\noindent Lastly, if $q - p - 1 \equiv 2$ (mod 8) then $q - p \equiv 3$ (mod 4) so that $\Cpq$ has two irreducible representations, where the psuedoscalar is $\pm 1$.  Thus we can represent it with $\{ e_1, e_2, \ldots, e_{p+q-1}, e_1 e_2 \ldots e_{p+q-1} \}$ where $\{ e_1, \ldots, e_{p+q-1} \}$ generate either $Cl_{p,q-1}$ or $Cl_{p-1,q}$ (both of which are type $\mathbb{H}$).  We can take $I$, $J$, and $K$ to be the same matrices as those in $Cl_{p,q-1}$ or $Cl_{p-1,q}$.  \\

\subsubsection{Real Representations of type $\C$ Algebras}
\noindent Suppose we have a surjective representation of $\Cpq$ on $\C^n$.  We can get a representation on $\R^{2n}$ by replacing $i$ with $\left(\begin{array}{cc}0 & -1 \\1 & 0\end{array}\right)$ and $1$ with $\left(\begin{array}{cc}1 & 0 \\0 & 1\end{array}\right)$.  By dimensionality considerations, $\Cpq$ is isomorphic to a proper subalgebra of $\R[2n]$ but the spin spaces have the same real dimension in each case. \\

\noindent If $\Cpq$ is of type $\C$ then $q - p - 1 \equiv 0$ or 4 (mod 8).  If $q - p - 1 \equiv 0$ then $q - (p+1) - 1 \equiv 7$ so that $Cl_{p+1,q}$ is of type $\R$.  If $q - p - 1 \equiv 4$ then $(q+1) - p - 1 \equiv 5$ so that $Cl_{p,q+1}$.  In either case, we get a real representation of $\Cpq$ by taking $\{e_1, e_2, \ldots, e_{p+q} \}$, where $\{ e_1, \ldots, e_{p+q+1} \}$ generate $Cl_{p+1,q}$ if $q - p - 1 \equiv 0$ and $Cl_{p, q-1}$ if $q - p - 1 \equiv 4$.  We will denote the extra generator, which anti-commutes with $\Rpq$, by $Z$.  We note that $\{e_1, \ldots, e_{p+q}, Z \}$ generates $End_\R(S)$. \\

\noindent Since $q - p \equiv 1$ (mod 4) and $(p+q)^2 \equiv 1$ (mod 4) ($p+q$ is odd), we have that $\gamma^2 = -1$ (\ref{g2}).  Furthermore, since $p+q$ is odd, $\gamma$ is in the center of $\Cpq$ (but anti-commutes with $Z$).  We give $\mathcal{S}$ a complex structure by defining
\begin{equation*}
(a+bi) \psi = (a + b \gamma) \psi, ~~ \psi \in \mathcal{S}.
\end{equation*}
If $q - p - 1 \equiv 0$ (mod 8) then $Z^2 = 1$ and, since $Z$ anti-commutes with $\Rpq$ and $p+q$ is odd, $Z$ anti-commutes with $i (= \gamma)$.  Therefore we can think of $Z$ as conjugation.  When $q - p -1 \equiv 4$ (mod 8) then $Z$ still anti-commutes with $i$ but we now have that $Z^2 = -1$.  It therefore seems tempting to give $\Cpq$ a quaternionic structure with $i, Z$, and $iZ$.  However, we hesitate to do this since $i$ and $Z$ are not interchangeable; $i$ is in the center of $\Cpq$ while $Z$ anti-commutes with $\Rpq$.

\subsubsection{Some Remarks}
\noindent If adding these additional operators for the type $\mathbb{H}$ and $\C$ algebras makes you uncomfortable, keep in mind that this is actually implicitly done in the standard Dirac treatment of $Cl_{1,3} \simeq \mathbb{H}[2]$, which uses a representation on $\C^4$.  In the Dirac theory, multiplying a spinor by $i$ is allowed but $i$ cannot be in the image of the representation of $Cl_{1,3}$ since the center of $Cl_{1,3}$ is $\{ \pm 1 \}$.  Indeed, the Dirac matrices and $i$ generate all of $\C[4]$.  It is interesting that there is a quaternionic structure if one uses a representation of $Cl_{1,3}$ on $\R^8$ but there is not if one uses a representation on $\C^4$.

\subsubsection{Trace-free Property of Real Representations}
\noindent An important feature of real representations is that all $e_I$ which are not $\pm 1$ are trace-free.  Since $tr(AB) = tr(BA)$, any matrix that is the product of anti-commuting matrices must be trace-free.  We will show that any $e_I$ can be written in such a way.  This is immediate if $|I|$ is even since the first $|I| - 1$ factors anti-commute with the last factor.  For odd $|I|$, first assume that $e_I \ne \gamma$ so that we can find $e_j$ such that $j \notin I$.  Then $e_I$ is proportional to $e_j e_I e_j$ and $e_j e_I$ and $e_j$ anti-commute.  Now if $|I|$ is odd and $e_I = \gamma \ne \pm 1$ then the Clifford algebra must be type $\C$ and $\gamma^2 = -1$.  However, this means that $\gamma^t = -\gamma$ so that $tr(\gamma) = tr(\gamma^t) = -tr(\gamma) = 0$.

\section{Bilinear Forms on Spinors}
Bilinear forms on spinors are discussed in \cite{lounesto} but from a different perspective.  Our approach is to look for real valued bilinear forms on the space of spinors, $S$, such that vectors are self-adjoint, up to sign.  That is, a bilinear function $(\cdot, \cdot): S \times S \to \R$ such that 
\begin{equation}
(\phi, v \psi) = \pm (v \phi, \psi) \label{sa}
\end{equation}
for all $v \in \Rpq$, $\phi,\psi \in S$.  If $\Cpq$ is type $\R$ then $\Cpq \simeq End_\R(S)$ so that the form can be represented as $(\phi, \psi) \mapsto \phi^t A \psi$, where $A \in \Cpq$ (where we are identifying an element $A \in End_\R(S)$ with its image under an isomorphism $End_\R(S) \to \Cpq$).  The condition (\ref{sa}) then becomes
\begin{equation*}
Av = \pm v^t A
\end{equation*}
for all $v \in \Rpq$.  If $1 \le i \le p$ then $e_i^t = e_i$ and if $p+1 \le i \le p+q$  then $e_i^t = -e_i$.  Therefore we must have that
\begin{equation*}
A e_i = \begin{cases}
\pm e_i A &\text{ if $1 \le i \le p$} \\
\mp e_i A &\text{ if $p+1 \le i \le p+q$.}
\end{cases}
\end{equation*}
Put
\begin{equation*}
A = \sum_{I \subseteq \{1, 2, \ldots, p+q \}} A_I e_I.
\end{equation*}
Since $e_i$ either commutes or anti-commutes with each $e_I$, if $A_I \ne 0$ then we must have that
\begin{equation*}
e_I e_i = \begin{cases}
\pm e_i e_I &\text{ if $1 \le i \le p$} \\
\mp e_i e_I &\text{ if $p+1 \le i \le p+q$.}
\end{cases}
\end{equation*}
It follows that if $A_I \ne 0$ then $e_I$ must be either $e_1 e_2 \ldots e_p$ or $e_{p+1} e_{p+2} \ldots e_{p+q}$.  Thus $A$ is, up to scale, either $e_1 e_2 \ldots e_p$ or $e_{p+1} e_{p+2} \ldots e_{p+q}$.  We denote the former element by $\gamma_p$ and the latter by $\gamma_q$.  We define two real bilinear forms
\begin{gather*}
(\phi, \psi)_+ = \phi^t \gamma_p \psi, \\
(\phi, \psi)_- = \phi^t \gamma_q \psi.
\end{gather*}
When $\Cpq$ is a corner algebra, i.e. $\gamma = \pm 1$, we have that $\gamma_p = \pm \gamma_q$.  Thus there is only one bilinear form, which we denote by $(\cdot, \cdot)$.  \\

\noindent We see that 
\begin{gather}
\gamma_p^2 = (e_1 e_2 \ldots e_p)(e_1 e_2 \ldots e_p) = (-1)^{(p-1) + (p-2) + \ldots + 1} e_1^2 e_2^2 \ldots e_p^2 \notag \\ = (-1)^{p(p-1)/2} \label{gp2}
\end{gather}
and, similarly,
\begin{equation}
\gamma_q^2 = (-1)^{q(q-1)/2} (-1)^q = (-1)^{q(q+1)/2}. \label{gq2}
\end{equation}
Thus $(\cdot, \cdot)_+$ is symmetric if $p = 0$ or 1 (mod 4) and anti-symmetric if $p=2$ or 3 (mod 4) and $(\cdot, \cdot)_-$ is symmetric if $q = 0$ or 3 (mod 4) and anti-symmetric if $p = 1$ or 2 (mod 4). \\

\noindent Given any vector $v \in \Rpq$, we can put $v = v_+ + v_-$, with $v_+ \in span\{e_1, e_2, \ldots e_p\}$ and $v_- \in span\{e_{p+1}, e_{p+2}, \ldots, e_{p+q}\}$.  We then have
\begin{align}
(\phi, v \psi)_+ &= \phi^\dagger \gamma_p (v_+ + v_-) \psi \nonumber \\ 
&= \phi^\dagger ((-1)^{p+1} v_+ + (-1)^p v_-) \gamma_p \psi \nonumber \\
&= \phi^\dagger((-1)^{p+1} v_+^\dagger + (-1)^{p+1} v_-^\dagger) \gamma_p \psi \nonumber \\
&= (-1)^{p+1}(v \phi, \psi)_+. \label{sa+}
\end{align}
A similar calculation yields
\begin{equation}
(\phi, v \psi)_- = (-1)^q (v \phi, \psi)_-. \label{sa-}
\end{equation}
If $\Cpq$ is of type $\mathbb{H}$ then from (\ref{tis-1}), $p^t = -p$ for $p \in span \{I, J, K\}$.  Further, any pure quaternion commutes with all of $\Cpq$.  We therefore have that
\begin{align*}
(p \phi, \psi)_\pm &= \phi^t p^t \gamma_{p(q)} \psi \\
&= -\phi^t \gamma_{p(q)} p \psi \\
&= -(\phi, p \psi)_\pm.  
\end{align*}
This means that
\begin{equation}
(q \phi, \psi)_\pm = (\phi, \bar{q} \psi)_\pm \text{ for all $q \in \mathbb{H}$}. \label{conjsym}
\end{equation}

\subsection{Signature of $(\cdot, \cdot)_\pm$}
We will now show that when $(\cdot, \cdot)_\pm$ is symmetric, the signature is always either neutral or definite.  The definite case occurs only when $p$ (respectively $q$) $= 0$.  Without loss of generality, we will show that $(\cdot, \cdot)_+$ has neutral signature when $\gamma_p$ is symmetric and $\gamma_p \ne 1$ (so that $p$ is necessarily non-zero).  Since $\gamma_p$ is symmetric, we have that $\gamma_p^2 = 1$.  Therefore its eigenvectors span the spinor space and its eigenvalues are $\pm 1$.  Let $S^\pm$ be the eigenspaces with eigenvalue $\pm 1$.  \\

\noindent First assume that $p$ is even or $q \ne 0$.  Let
\begin{equation*}
e = \begin{cases}
e_1 &\text{ if $p$ is even} \\
e_{p+1} &\text{ if $p$ is odd}.
\end{cases}
\end{equation*}
We have that $e \gamma_p = -\gamma_p e$ so that if $\psi \in S^+$ then $e \psi \in S^-$.  Since $e$ is invertible it follows that $S^+$ and $S^-$ have the same dimension so that $(\cdot, \cdot)_+$ has neutral signature. \\

\noindent Assume now that $p$ is odd and $q = 0$.  If $\Cpq = Cl_{p,0}$ is of type $\R$ or $\mathbb{H}$ then we have that $\gamma_p = \gamma = \pm 1$, so that $(\cdot, \cdot)_+$ is the Euclidean inner product.  On the other hand, if $\Cpq$ is type $\C$ then $\gamma_p^2 = \gamma^2 = -1$ so that $(\cdot, \cdot)_+$ is anti-symmetric.

\subsection{$Pin$ and $Spin$}
Recall that the action of the $Pin$ and $Spin$ groups on vectors preserves the metric.  We also see that the action of the $Pin$ and $Spin$ groups on spinors (which is just left multiplication) preserves $(\cdot, \cdot)_\pm$ up to sign:
\begin{equation*}
(u \phi, u \psi)_\pm = \pm (\phi, \psi)_\pm, 
\end{equation*}
for all $u \in Pin(p,q), \phi,\psi \in S$.  This is evident from (\ref{sa+}), (\ref{sa-}), and the fact that for $u \in Pin(p,q), u\tilde{u} = \pm 1$.  We can define a subgroup $Pin_+(p,q)$  of $Pin(p,q)$ by
\begin{equation*}
Pin_+(p,q) = \{ u \in Pin(p,q) : u\tilde{u} = 1 \}.
\end{equation*}
Then we see that the action of $Pin_+(p,q)$ on spinors preserves $(\cdot, \cdot)_+$ if $p$ is odd and preserves $(\cdot, \cdot)_-$ if $q$ is even.  Furthermore, $Spin_+(p,q) = Pin_+(p,q) \cap Spin(p,q)$ always preserves $(\cdot, \cdot)_\pm$.

\subsection{Relations Between Spinors and Forms}
Denote the exterior algebra on $\Rpq$ by $\Lambda(\Rpq)$.  The metric $g$ on $\Rpq$ induces a metric on $\Lambda(\Rpq)$ by
\begin{equation*}
g(e_{i_1} e_{i_2} \ldots e_{i_m}, e_{j_1} e_{j_2} \ldots e_{j_n}) = \prod_{i_k = j_l} g(e_{i_k},e_{j_l}),
\end{equation*}
where $i_k < i_{k+1}$ and $j_k < j_{k+1}$.  If $\Cpq$ is a corner algebra and $q$ is odd, then there is an ambiguity because of self-duality.  Recall that in a corner algebra $e_1 \ldots e_{p+q} = \pm 1$.  Therefore $e_I = \pm e_{I^c}$ (where $I^c$ is the complement of $I$).  But $g(e_I, e_I) = -g(e_{I^c}, e_{I^c})$ since the number of unit vectors which square to -1 in $e_I$ will have the opposite parity of the amount in $e_{I^c}$.  To resolve this problem, we cut off forms at degree $\frac{p+q-1}{2}$.  That is, we consider $\Lambda(\Rpq)$ to be $\bigoplus_{k=0}^{(p+q-1)/2} \Lambda^k(\Rpq)$. \\

\noindent We can use the bilinear forms to associate elements $\om_k^\pm(\phi, \psi)$ of the dual space of $\Lambda^k(\Rpq)$ with spinors $\phi, \psi$ by
\begin{equation*}
\om_k^\pm(\phi, \psi)(u) = (\phi, u \psi)_\pm, u \in \Lambda^k(\Rpq).
\end{equation*}
The induced metric on $\Lambda^k(\Rpq)$ allows us to consider $\om_k^\pm(\phi, \psi)$ to be an element in $\Lambda^k(\Rpq)$. \\

\noindent We then have that
\begin{equation*}
g(\om_k^\pm(\phi,\psi), \om_k^\pm(\phi, \psi)) = \sum_{|I| = k} g(e_I, e_I) (\phi, e_I \psi)_\pm^2.
\end{equation*}

\noindent Because of vectors being self-adjoint (up to sign), we see that the function $\om_k^\pm : S \times S \to \Lambda^k(\Rpq)$ is either symmetric or anti-symmetric in its two spinor arguments.  More precisely, by (\ref{gp2}), (\ref{gq2}), (\ref{sa+}), and (\ref{sa-}) we have that
\begin{align*}
(\psi, e_{i_1} \ldots e_{i_k} \phi)_+ &= (-1)^{k(p+1)}(e_{i_k} \ldots e_{i_1} \psi, \phi)_+ \\
&= (-1)^{k(p+1) + k(k-1)/2}(e_{i_1} \ldots e_{i_k} \psi, \phi)_+ \\
&= (-1)^{k(p+1) + k(k-1)/2 + p(p-1)/2}(\phi, e_{i_1} \ldots e_{i_k} \psi)_+ \\
&= (-1)^{\frac{1}{2}((k+p)^2 + k - p)}.
\end{align*}
Thus
\begin{align}
\om_k^+(\psi, \phi) &= (-1)^{\frac{1}{2}((k+p)^2 + k - p)} \om_k^+(\phi, \psi) \label{om+sym} \\
&= \om_k^+(\phi, \psi) \begin{cases} \notag
(-1)^{k(k-3)/2} &\text{ if $p \equiv 0$} \\
(-1)^{k(k-1)/2} &\text{ if $p \equiv 1$} \\
(-1)^{(k-1)(k-2)/2} &\text{ if $p \equiv 2$} \\
(-1)^{(k-2)(k-3)/2} &\text{ if $p \equiv 3$}.
\end{cases}
\end{align}
A similar calculation gives
\begin{align}
\om_k^-(\psi, \phi) &= (-1)^{\frac{1}{2}((k+q)^2+q-k)} \om_k^-(\phi, \psi) \label{om-sym} \\
&= \om_k^-(\phi, \psi) \begin{cases} \notag
(-1)^{k(k-1)/2} &\text{ if $q \equiv 0$} \\
(-1)^{(k-1)(k-2)/2} &\text{ if $q \equiv 1$} \\
(-1)^{(k-2)(k-3)/2} &\text{ if $q \equiv 2$} \\
(-1)^{k(k-3)/2} &\text{ if $q \equiv 3$}.
\end{cases}
\end{align}
These imply that
\begin{align}
\om_k^+(\phi,\phi) \ne 0 \text{ for all $\phi$ if and only if $(k+p)^2 - p + k \equiv 0$ (mod 4)} \label{vpp+} \\
\om_k^-(\phi,\phi) \ne 0 \text{ for all $\phi$ if and only if $(k+q)^2 + q - k \equiv 0$ (mod 4)} \label{vpp-}.
\end{align}
Note that because of (\ref{conjsym}), for any pure quaternion $q$ we have that
\begin{equation}
\om_k^\pm(\phi, \phi) = 0 \text{ for all $\phi$ if and only if } \om_k^\pm(\phi, q \phi) \ne 0 \text{ for all $\phi$}. \label{vppH}
\end{equation}

\section{Derivation of Identities}
\noindent Motivated by the results obtained in the appendix of $\cite{penroserindler}$, the starting point for finding identities is the consideration of two rank (2,2) spinorial tensors $T_\pm$ defined by
\begin{equation}
{T_\pm}^{\mu \nu}_{\rho \sigma} = \sum_{I \subseteq \{1, \ldots, p+q\}} (\pm 1)^{|I|} g(e_I, e_I) {e_I}_\rho^\mu {e_I}_\sigma^\nu. \label{defT}
\end{equation}
A more specialized version of the above tensor is used in \cite{mtheory} as the starting point for their Fierz identities.  Let $\langle \cdot, \cdot \rangle = (\cdot, \cdot)_+$ or $(\cdot, \cdot)_-$ and $\Gamma = \gamma_p$ or $\gamma_q$ accordingly.  The bilinear form allows us to lower indices via $\psi_\nu = \Gamma_{\mu \nu} \psi^\mu$.  The inverse metric is defined by $\Gamma^{\mu \nu} \phi_\nu = \phi^\mu$.  Since $\gamma_{p(q)}^{-1} = \gamma_{p(q)}^t$, the matrix forms of $\Gamma^{\mu \nu}$ and $\Gamma_{\mu \nu}$ are the same.  From (\ref{defT}) we have
\begin{equation*}
{T_\pm}_{\kappa \rho \theta \sigma} = \sum_{I \subseteq \{1, \ldots, p+q\}} (\pm 1)^{|I|} g(e_I, e_I) {e_I}_\rho^\mu {e_I}_\sigma^\nu \Gamma_{\kappa \mu} \Gamma_{\theta \nu}
\end{equation*}
so that
\begin{align}
{T_\pm}_{\kappa \rho \theta \sigma} \phi^\kappa \psi^\rho \alpha^\theta \beta^\sigma &= \sum_{I \subseteq \{1, \ldots, p+q\}} (\pm 1)^{|I|} g(e_I, e_I) \langle \phi, e_I \psi \rangle \langle \alpha, e_I \beta \rangle \nonumber \\
&= \sum_{k=0}^{p+q} (\pm 1)^k \om_k(\phi, \psi) \cdot \om_k(\alpha, \beta), \label{Tandomega}
\end{align}
where $\om_k = \om_k^+$ or $\om_k^-$ according to whether $\Gamma = \gamma_p$ or $\gamma_q$.  An immediate symmetry relation is
\begin{equation}
{T_\pm}_{\kappa \rho \theta \sigma} = {T_\pm}_{\theta \sigma \kappa \rho}. \label{sym1}
\end{equation}
Further, since $\om_k(\psi, \phi) = \pm \om_k(\phi, \psi)$, we have that
\begin{equation}
{T_\pm}_{\kappa \rho \theta \sigma} = {T_\pm}_{\rho \kappa \sigma \theta} \label{sym2}.
\end{equation}

\noindent Define
\begin{equation*}
T_\pm(\phi, \psi, \alpha, \beta) = {T_\pm}_{\kappa \rho \theta \sigma} \phi^\kappa \psi^\rho \alpha^\theta \beta^\sigma.
\end{equation*}

\noindent An important fact is that the set $\{ \pm 1, \pm e_I : I \subseteq \{1, 2, \ldots, p+q \} \}$ forms a multiplicative group $G$ (of order $2^{p+q+1}$ if $p+q$ is even and $2^{p+q}$ if $p+q$ is odd).  Since every term in $T_\pm$ is quadratic in the $e_I$'s, it is natural to consider the quotient group $E = G/\{\pm 1\}$.  Note that $E$ is an abelian group of order $\frac{1}{2} |G|$ and every (non-identity) element has order two.  For the rest of this section, we identify $e_I$ with its image under the projection map $\pi : G \to E$.  The product of two elements of $E$ is then $e_I e_J = e_{I \Delta J}$ where $\Delta$ is the symmetric difference, i.e. $I \Delta J = (I \cup J) \backslash (I \cap J)$.  Indeed, $E$ is isomorphic to the group one gets when the set is the power set of $\{1, \ldots, p+q \}$ and the operation is $\Delta$. \\

\noindent We can write
\begin{equation*}
T_\pm(\phi, \psi, \alpha, \beta) = \sum_{e_I \in E} (\pm 1)^{|I|} g(e_I, e_I)\langle\phi, e_I \psi\rangle \langle \alpha, e_I \beta \rangle
\end{equation*}
so that for $e_J \in E$
\begin{align*}
T_\pm(\phi, e_J \psi, \alpha, e_J \beta) &= \sum_{e_I \in E} (\pm 1)^{|I|} g(e_I,e_I) \langle\phi, e_I e_J \psi\rangle \langle\alpha e_I e_J \beta\rangle \\
&= \sum_{e_I \in E} (\pm 1)^{|I \Delta J|} g(e_I e_J^{-1}, e_I e_J^{-1}) \langle\phi, e_I \psi\rangle \langle\alpha, e_I \beta \rangle \\
&= \sum_{e_I \in E} (\pm 1)^{|I| + |J|} g(e_J, e_J) g(e_I, e_I) \langle\phi, e_I \psi\rangle \langle\alpha, e_I \beta \rangle \\
&= (\pm 1)^{|J|} g(e_J, e_J) T(\phi, \psi, \alpha, \beta),
\end{align*}
where the second equality follows from $E$ being a group and the third from $| I \Delta J| = | I \cup J | - | I \cap J |  = |I| + |J| - 2|I \cap J|$ so that $(\pm 1)^{|I \Delta J|} = (\pm 1)^{|I| + |J|}$.  By (\ref{sym2}) we have
\begin{equation*}
T_\pm (e_J \phi, \psi, e_J \alpha, \beta) = (\pm 1)^{|J|} g(e_J, e_J) T_\pm(\phi, \psi, \alpha, \beta).
\end{equation*}

\noindent In particular, we have
\begin{subequations} \label{Te}
\begin{align}
T_\pm(\phi, e_i \psi, \alpha, e_i \beta) &= \pm e_i^2 T_\pm(\phi, \psi, \alpha, \beta) \\
T_\pm(e_i \phi, \psi, e_i \alpha, \beta) &= \pm e_i^2 T_\pm(\phi, \psi, \alpha, \beta).
\end{align}
\end{subequations}

\subsection{Type $\R$ Subordinate Algebras}
\noindent We will first assume that $\Cpq$ is a type $\R$ subordinate algebra, so that $\{ e_I : I \subseteq \{1, \ldots, p+q\} \}$ is a basis for $End_\R(\mathcal{S})$.  We then have that $\{ e_I \otimes e_J : I, J \subseteq \{1, \ldots, p+q\} \}$ is a basis for $\mathcal{S} \otimes \mathcal{S} \otimes \mathcal{S} \otimes \mathcal{S}$ where the action of $e_I \otimes e_J$ on four spinors is $(\phi, \psi, \alpha, \beta) \mapsto \langle \phi, e_I \alpha \rangle \langle \psi, e_J \beta \rangle$.  It turns out that the symmetries given in (\ref{Te}) are sufficient to find the expansion of $T_\pm$ in the basis that groups the first and third arguments, and the second and fourth.  Once we have this, we can re-raise the two lowered indices to get a factorization of the original (2,2) tensor in (\ref{defT}), which groups the two spinors and the two cospinors.  \\

\noindent Thus we want to determine coefficients $t_{I,J}^\pm \in \R$ so that we can write $T_\pm$ in this basis,
\begin{equation}
T_\pm (\phi, \psi, \alpha, \beta) = \sum_{I, J \subseteq \{1, \ldots, p+q\}} t_{I,J}^\pm \langle \phi, e_I \alpha \rangle \langle \psi, e_J \beta \rangle. \label{Tnewbasis}
\end{equation}
For concreteness, we will assume that $p$ is odd and that $\Gamma = \gamma_p$.  We first consider $T_+$.  Since $p$ is odd, we have that $(\phi, v \psi)_+ = (v \phi, \psi)_+$ (\ref{sa+}) for all $v \in \Rpq$.  From (\ref{Te}a) and (\ref{Tnewbasis}) we have that
\begin{align*}
\sum_{I, J} t^+_{I, J} (\phi, e_I \alpha)_+ (e_i \psi, e_J e_i \beta)_+ = e_i^2 \sum_{I, J} t^+_{I, J} (\phi, e_I \alpha)_+ (\psi, e_J \alpha)_+ \\
\Rightarrow \sum_{I, J} t^+_{I, J} (\phi, e_I \alpha)_+ (\psi, e_i e_J e_i \beta)_+ = e_i^2 \sum_{I, J} t^+_{I, J} (\phi, e_I \alpha)_+ (\psi, e_J \alpha)_+.
\end{align*}
Because $e_i e_J e_i = \pm e_J$ and $(\cdot, e_I \cdot)_+ (\cdot, e_J \cdot)_+$ is a basis for $S \otimes S \otimes S \otimes S$, we must have that
\begin{align}
e_i e_J e_i = e_i^2 e_J \nonumber \\
\Rightarrow e_J e_i = e_i e_J \label{comrelT+}
\end{align}
whenever $t^+_{I, J}$ is non-zero for some $I$.  Since $p+q$ is even, the only element in the center of $\Cpq$ is the identity.  Thus $e_J = 1$.  After using the same argument using the symmetry given in (\ref{Te}b), we see that
\begin{equation*}
T_+(\phi, \psi, \alpha, \beta) = C(\phi, \alpha)_+ (\psi, \beta)_+
\end{equation*}
for some $C \in \R$.  Since the factorization of $T_-$ in (\ref{Te}) differs from that of $T_+$ by a negative sign, the commutation relation for $T_-$ differs from (\ref{comrelT+}) by a minus sign:
\begin{equation*}
e_J e_i = -e_i e_J \text{ if $t^-_{I,J} \ne 0$ for some $I$}.
\end{equation*}
Since $\gamma$ is the only element which anti-commutes with every vector, we must have that
\begin{equation*}
T_-(\phi,\psi,\alpha,\beta) = C (\phi, \gamma \alpha)_+ (\psi, \gamma \beta)_+ = C (\phi, \alpha)_- (\psi, \beta)_-.
\end{equation*}

\noindent Analogous computations can be done to find the factorization of $T_\pm$ for different parities of $p$ and different choices of $\Gamma$.  In some cases $e_J$ must be in the center (so it must be $1$) and in others it must anti-commute with all vectors (so it must be $\gamma$).  We summarize the results in tables \ref{T+table} and \ref{T-table}.  Note that in these tables $C$ is not uniform across entries. \\
\begin{table}[!h]
\begin{center}
\begin{tabular}{|c|c|c|}
\hline
\backslashbox[0.5cm]{$~~~~p$}{$\Gamma$} & $\gamma_p$ & $\gamma_q$ \\ \hline
even & $C (\phi, \alpha)_- (\psi, \beta)_-$ & $C (\phi, \alpha)_- (\psi, \beta)_-$ \\ \hline
odd & $C (\phi, \alpha)_+ (\psi, \beta)_+$ & $C (\phi, \alpha)_+ (\psi, \beta)_+$ \\ \hline
\end{tabular}
\caption{Factorization of $T_+(\phi,\psi,\alpha,\beta)$.} \label{T+table}
\end{center}
\end{table}

\begin{table}[!h]
\begin{center}
\begin{tabular}{|c|c|c|}
\hline
\backslashbox[0.5cm]{$~~~~p$}{$\Gamma$} & $\gamma_p$ & $\gamma_q$ \\ \hline
even & $C (\phi, \alpha)_+ (\psi, \beta)_+$ & $C (\phi, \alpha)_+ (\psi, \beta)_+$ \\ \hline
odd & $C (\phi, \alpha)_- (\psi, \beta)_-$ & $C (\phi, \alpha)_- (\psi, \beta)_-$ \\ \hline
\end{tabular}
\caption{Factorization of $T_-(\phi,\psi,\alpha,\beta)$.} \label{T-table}
\end{center}
\end{table}

\noindent To determine the value of $C$ we will need to raise two indices to recover ${T_\pm}_{\rho \sigma}^{\mu \nu}$.  This process is slightly different for the cases where the factorization of $T_\pm$ uses $\Gamma$ as the inner product (e.g. when $p$ is even and using $T_+$ and $\Gamma = \gamma_p$) and those where the factorization uses $\gamma \Gamma$ as the inner product (e.g. when $p$ is even and using $T_+$ and $\Gamma = \gamma_q$). \\

\noindent In the first case, we have (in indices)
\begin{equation*}
{T_\pm}_{\kappa \rho \theta \sigma} = C \Gamma_{\kappa \theta} \Gamma_{\rho \sigma}
\end{equation*}
so that
\begin{align}
{T_\pm}_{\rho \sigma}^{\mu \nu} &= C \Gamma_{\kappa \theta} \Gamma_{\rho \sigma} \Gamma^{\mu \kappa} \Gamma^{\nu \theta} \notag \\
&= C \Gamma_{\kappa \theta} \Gamma^{\nu \theta} \Gamma_{\rho \sigma} \Gamma^{\mu \kappa} \notag \\
&= C \delta_\kappa^\nu \Gamma_{\rho \sigma} \Gamma^{\mu \kappa} \notag \\
&= C \Gamma_{\rho \sigma} \Gamma^{\mu \nu}, \label{T22fac}
\end{align}
the second to last inequality coming from the fact that $\Gamma^t \Gamma = 1$.  To determine $C$ we consider the double trace ${T_\pm}_{\mu \nu}^{\mu \nu}$.  From (\ref{defT}) we have
\begin{equation*}
{T_\pm}_{\mu \nu}^{\mu \nu} = \sum_{I} g(e_I, e_I) {e_I}_\mu^\mu {e_I}_\nu^\nu.
\end{equation*}
But, as discussed in section 1, all $e_I$ are trace-free except for $e_\emptyset = 1$.  Thus
\begin{equation*}
{T_\pm}_{\mu \nu}^{\mu \nu} = tr(1)^2 = 2^{p+q}.
\end{equation*}
From (\ref{T22fac}), we see that ${T_\pm}_{\mu \nu}^{\mu \nu}$ is $tr(C \gamma_p^t \gamma_p) = C tr(1) = C 2^{(p+q)/2}$.  Thus $C = 2^{(p+q)/2}$. \\

\noindent  Consider now the case where the factorization of $T_\pm$ uses the bilinear form induced by $\gamma \Gamma$.  The cases $\Gamma = \gamma_p$ and $\Gamma = \gamma_q$ are slightly different but analogous.  Consider the case where $\Gamma = \gamma_p$ so that we have
\begin{equation*}
{T_\pm}_{\kappa \rho \theta \sigma} = C {\gamma_q}_{\kappa \theta} {\gamma_q}_{\rho \sigma}.
\end{equation*}
Raising two indices we have
\begin{align}
{T_\pm}_{\rho \sigma}^{\mu \nu} &= C {\gamma_q}_{\kappa \theta} {\gamma_q}_{\rho \sigma} {\gamma_p}^{\mu \kappa} {\gamma_p}^{\nu \theta} \notag \\
&= C {\gamma_q}_{\kappa \theta} {\gamma_p}^{\nu \theta} {\gamma_p}^{\mu \kappa} {\gamma_q}_{\rho \sigma} \notag \\
&= C \epsilon_1 \gamma_\kappa^\nu {\gamma_p}^{\mu \kappa} {\gamma_q}_{\rho \sigma} \notag \\
&= C \epsilon_1 \epsilon_2 {\gamma_q}^{\mu \nu} {\gamma_q}_{\rho \sigma},
\end{align}
where $\epsilon_i$ are such that
\begin{align*}
\gamma_q \gamma_p^t &= \epsilon_1 \gamma \\
\gamma_p \gamma &= \epsilon_2 \gamma_q.
\end{align*}
We have that
\begin{align*}
\gamma_q \gamma_p^t &= e_{p+1} \ldots e_{p+q} e_p^t \ldots e_1^t \\
&= e_{p+1} \ldots e_{p+q} e_p \ldots e_1 \\
&= (-1)^{p(p-1)/2} e_{p+1} \ldots e_{p+q} e_1 \ldots e_p \\
&= (-1)^{p(p-1)/2 + pq} \gamma 
\end{align*}
and
\begin{align*}
\gamma_p \gamma &= e_1 \ldots e_p e_1 \ldots e_{p+q} \\
&= (-1)^{p(p-1)/2} e_1^2 \ldots e_p^2 e_{p+1} \ldots e_{p+q} \\
&= (-1)^{p(p-1)/2} \gamma_q.
\end{align*}
Thus
\begin{equation*}
\epsilon_1 \epsilon_2 = (-1)^{pq}.
\end{equation*}
Considering ${T_\pm}_{\mu \nu}^{\mu \nu}$, we find that $C = (-1)^{pq} 2^{(p+q)/2}$.  When $\Gamma = \gamma_q$, a similar argument shows that $C = (-1)^{((p+q)^2 + q(q+2) - p^2)/2} 2^{(p+q)/2}$.  Since in a subordinate algebra $p+q$ is even, this simplifies to $C = (-1)^{(q(q+2) - p^2)/2} 2^{(p+q)/2}$.  We now can update tables \ref{T+table} and \ref{T-table} to get tables \ref{T+tablenew} and \ref{T-tablenew}. \\

\begin{table}[!h]
\begin{center}
\begin{tabular}{|c|c|c|}
\hline
\backslashbox[0.5cm]{$~~~~p$}{$\Gamma$} & $\gamma_p$ & $\gamma_q$ \\ \hline
even & $2^{(p+q)/2} (\phi, \alpha)_- (\psi, \beta)_-$ & $2^{(p+q)/2} (\phi, \alpha)_- (\psi, \beta)_-$ \\ \hline
odd & $2^{(p+q)/2} (\phi, \alpha)_+ (\psi, \beta)_+$ & $-2^{(p+q)/2} (\phi, \alpha)_+ (\psi, \beta)_+$ \\ \hline
\end{tabular}
\caption{Factorization of $T_+(\phi,\psi,\alpha,\beta)$ for type $\R$ subordinate algebras.} \label{T+tablenew}
\end{center}
\end{table}

\begin{table}[!h]
\begin{center}
\begin{tabular}{|c|c|c|}
\hline
\backslashbox[0.5cm]{$~~~~p$}{$\Gamma$} & $\gamma_p$ & $\gamma_q$ \\ \hline
even & $2^{(p+q)/2} (\phi, \alpha)_+ (\psi, \beta)_+$ & $2^{(p+q)/2} (\phi, \alpha)_+ (\psi, \beta)_+$ \\ \hline
odd & $-2^{(p+q)/2} (\phi, \alpha)_- (\psi, \beta)_-$ & $2^{(p+q)/2} (\phi, \alpha)_- (\psi, \beta)_-$ \\ \hline
\end{tabular}
\caption{Factorization of $T_-(\phi,\psi,\alpha,\beta)$ for type $\R$ subordinate algebras.} \label{T-tablenew}
\end{center}
\end{table}

\noindent It is somewhat surprising that the factorization of ${T_\pm}_{\rho \sigma}^{\mu \nu}$ is either ${\gamma_p}_{\rho \sigma} {\gamma_p}^{\mu \nu}$ or ${\gamma_q}_{\rho \sigma} {\gamma_q}^{\mu \nu}$, despite the definition of ${T_\pm}_{\rho \sigma}^{\mu \nu}$ (\ref{defT}) not singling out $\gamma_p$ or $\gamma_q$.  Since $T_\pm$ is a fruitful object (its complex version is used in \cite{penroserindler} and a more specialized version is used for Fierz identities in \cite{mtheory}) this is seen as a hint that the bilinear forms induced by $\gamma_p$ and $\gamma_q$ are important structures. \\

\subsubsection{Cross-Symmetry} \label{Rsubcs}
The factorization of $T_\pm$ alone gives many Fierz identities by specializing, for example, to $\psi = \phi$ and $\alpha = \beta$ and noticing that many terms vanish because of (\ref{vpp+}) and (\ref{vpp-}).  However, we will here derive more four-spinor symmetries, which we call \textit{cross-symmetry} and which are encoded in certain involutory matrices.  These matrices along with tables \ref{T+tablenew} and \ref{T-tablenew}, contain all of the information for the Fierz identities we will derive.  We can look at tables \ref{T+tablenew} and \ref{T-tablenew} as giving $\om_0^\pm(\phi, \alpha) \cdot \om_0^\pm (\psi, \beta)$ in terms of $\om_k^\pm(\phi, \psi) \cdot \om_k^\pm(\alpha, \beta)$.  The cross-symmetry relations give us $\om_m^\pm(\alpha, \psi) \cdot \om_m^\pm (\phi, \beta)$ in terms of $\om_k^\pm(\phi, \psi) \cdot \om_k^\pm(\alpha, \beta)$, for arbitrary $m$. \\

\noindent From tables \ref{T+tablenew} and \ref{T-tablenew}, we see that given any subordinate type $\R$ algebra and choice of $\Gamma$, exactly one of $T_+$ or $T_-$ uses the bilinear form induced by $\Gamma$ in its factorization.  Which of $T_\pm$ to use for a given partity of $p$ and $\Gamma$ is shown in table \ref{fac=gam}. \\

\begin{table}[!h]
\begin{center}
\begin{tabular}{|c|c|c|}
\hline
\backslashbox[0.5cm]{$~~~~p$}{$\Gamma$} & $\gamma_p$ & $\gamma_q$ \\ \hline
even & $T_-$ & $T_+$ \\ \hline
odd & $T_+$ & $T_-$ \\ \hline
\end{tabular}
\caption{Which of $T_\pm$ uses $\Gamma$ in its factorization.} \label{fac=gam}
\end{center}
\end{table}

\noindent It turns out that the computations that follow \textit{do not} depend on the parity of $p$ and the choice of $\Gamma$.  However, to show computations explicitly, we first assume that $p$ (and therefore $q$) is even and we use $T_+$.  From (\ref{Tandomega}) and table \ref{T+tablenew} we have (noting that $(\phi, \alpha)_- (\psi, \beta)_- = \om_0^-(\phi, \alpha) \cdot \om_0^-(\psi,\beta)$)
\begin{equation*}
2^{(p+q)/2} \om_0^-(\phi, \alpha) \cdot \om_0^- (\psi, \beta) = \sum_{k=0}^{p+q} \om_k^-(\phi, \psi) \cdot \om_k^-(\alpha, \beta).
\end{equation*}
From $\ref{om-sym}$, interchanging $\phi$ and $\psi$ in the above equation and then rewriting $\om_0^-(\psi, \alpha)$ as $(-1)^{\frac{1}{2}(q^2 + q)} \om_0(\alpha, \psi)$ gives
\begin{gather}
(-1)^{\frac{1}{2}(q^2 + q)} 2^{n/2} \om_0^-(\alpha, \psi) \cdot \om_0^-(\phi, \beta) = \sum_{k=0}^{p+q} (-1)^{\frac{1}{2}((k+q)^2 + q - k)} \om_k^-(\phi, \psi) \cdot \om_k^-(\alpha, \beta) \notag \\
\Rightarrow  \om_0^-(\alpha, \psi) \cdot \om_0(\phi, \beta) = \sum_{k=0}^{p+q} \eps_k \om_k^-(\phi, \psi) \cdot \om_k^-(\alpha, \beta) \label{startpt}
\end{gather}
where
\begin{align}
\eps_k &= \frac{1}{2^{(p+q)/2}} (-1)^{\frac{1}{2}k(k - 1) + qk} \notag \\
&= \frac{1}{2^{(p+q)/2}} (-1)^{\frac{1}{2}k(k-1)} \text{ (since $q$ is even).} \label{epsdef}
\end{align}
It turns out (surprisingly) that $\eps_k$ is the same, regardless of the parity of $p$ and the choice of $T_\pm$ and $\Gamma$.  From now on, we will therefore work in general with $\langle \cdot, \cdot \rangle$ representing $(\cdot, \cdot)_\pm$ depending on whether $\Gamma = \gamma_p$ or $\gamma_q$.  Similarly, we write $\om$ for $\om^\pm$ and $T$ or $T_\pm$.  \\



\noindent Recall that $T$ is invariant (up to sign) under $(\psi, \beta) \mapsto (e_J \psi, e_J \beta)$.  This is not the case in (\ref{startpt}) since $\eps_k$ can be different for values of $k$ of the same parity.  Indeed, rewriting the general version of (\ref{startpt}) as
\begin{equation*}
\langle \alpha, \psi \rangle \langle \phi, \beta \rangle = \sum_{|I| \le p+q} \eps_k g(e_I, e_I) \langle \phi, e_I \psi \rangle \langle \alpha, e_I \beta \rangle,
\end{equation*}
replacing $(\psi, \beta) \mapsto (e_J \psi, e_J \beta)$ and multiplying by $g(e_J, e_J)$ we have
\begin{align}
g(e_J, e_J) \langle \alpha, e_J \psi \rangle \langle \phi, e_J \beta \rangle &= \sum_{|I| \le n} \eps_k g(e_J, e_J) g(e_I, e_I) \langle \phi, e_I e_J \psi \rangle \langle \alpha, e_I e_J \beta \rangle \notag \\
g(e_J, e_J) \langle \alpha, e_J \psi \rangle \langle \phi, e_J \beta \rangle &= \sum_{|I| \le n} \eps_{|I \Delta J|} g(e_I, e_I) \langle \phi, e_I \psi \rangle \langle \alpha, e_I \beta \rangle. \label{eJ}
\end{align}
In the last line we appealed to the group structure of $\{ e_I : I \subset \{1, \ldots, n\} \} / \{ \pm 1 \}$.  It follows from from the definition of $\eps_k$ (\ref{epsdef}) that $\eps_k = -\eps_{k+2}$, so that terms of $\om_k(\phi, \psi) \cdot \om_k(\alpha, \beta)$ come in with different signs.  More explicitly, since $|I \Delta J| = |I| + |J| - 2|I \cap J|$ and $\eps_k$ depends only on $k$ mod 4, terms $g(e_I, e_I) (\phi, e_I \psi) (\alpha, e_I \beta)$ and $g(e_{I'},e_{I'}) (\phi, e_{I'} \psi) (\alpha, e_{I'} \beta)$ of $\om_k(\phi, \psi) \cdot \om_k(\alpha, \beta)$ will have the same sign if and only if $|I \cap J| = |I' \cap J|$. \\

\noindent If we sum (\ref{eJ}) over for all $J$ of a fixed size, $j$, then we must get something invariant since the left hand side will be equal to $\om_j(\alpha, \psi) \cdot \om_j(\phi, \beta)$.  Finding what the right hand side will be equal to is a little less trivial.  For a fixed $I \subset \{1, \ldots, p+q \}$, consider how many times, and in what sign, we can get $g(e_I, e_I)\langle \phi, e_I \psi \rangle \langle \alpha, e_I \beta \rangle$.  Put $k = |I|$.  For fixed $m$, $0 \le m \le \min\{j, k \}$ we consider how many $J$ there are such $|I \cap J| = m$.  To do this we first choose $m$ elements from $I$ and then $j - m$ elements from the $n - k$ elements which are in the complement of $I$.  Therefore the coefficient on $g(e_I, e_I)\langle \phi, e_I \psi \rangle \langle \alpha, e_I \beta \rangle$ is $\eps_{k + j - 2m} {k \choose m} {p+q - k \choose j - m} = \eps_{k+j} (-1)^m {k \choose m}{p+q - k \choose j - m}$.  We thus have the following general formula:
\begin{equation}
\om_j(\alpha, \psi) \cdot \om_j(\phi, \beta) = \sum_{k=0}^{p+q} \eps_{k+j} \sum_{m=0}^{\min\{j,k\}} (-1)^m {k \choose m} {p+q-k \choose j-m} \om_k(\phi, \psi) \cdot \om_k(\alpha, \beta). \label{genform}
\end{equation}
Define maps $\Omega$ and $\Omega^*$ from $\mathcal{S} \times \mathcal{S} \times \mathcal{S} \times \mathcal{S}$ into $\R^{p+q+1}$ by
\begin{align}
\Omega &: (\phi, \psi, \alpha, \beta) \mapsto \left(\begin{array}{c} \om_0(\phi, \psi) \cdot \om_0(\alpha, \beta) \\ \vdots \\ \om_{p+q}(\phi, \psi) \cdot \om_{p+q}(\alpha, \beta) \end{array}\right)\label{omdef} \\
\Omega^* &: (\phi, \psi, \alpha, \beta) \mapsto \left(\begin{array}{c} \om_0(\alpha, \psi) \cdot \om_0(\phi, \beta) \\ \vdots \\ \om_{p+q}(\alpha, \psi) \cdot \om_{p+q}(\phi, \beta) \end{array}\right). \label{om*def}
\end{align}
These maps are clearly linear in each variable so they induce (unique) linear maps from $\mathcal{S} \otimes \mathcal{S} \otimes \mathcal{S} \otimes \mathcal{S} \to \R^{p+q+1}$.  If we let $M$ be the $p+q+1 \times p+q+1$ matrix whose $j+1, k+1$ component is 
\begin{equation}
\eps_{k+j} \sum_{m=0}^{\min\{j,k\}} (-1)^m {k \choose m} {p+q-k \choose j-m} \label{Mdef}
\end{equation}  
then $(3)$ tells us that
\begin{equation*}
\Omega^*(\phi \otimes \psi \otimes \alpha \otimes \beta) = M \Omega(\phi \otimes \psi \otimes \alpha \otimes \beta).
\end{equation*}
Furthermore, from the definitions of $\Omega$ and $\Omega^*$, we clearly have that $M^2 \Omega(\phi \otimes \psi \otimes \alpha \otimes \beta) = \Omega (\phi \otimes \psi \otimes \alpha \otimes \beta)$.  Thus if we can show that the map $\Omega : \mathcal{S} \otimes \mathcal{S} \otimes \mathcal{S} \otimes \mathcal{S} \to \R^{p+q+1}$ is surjective, then it will follow that $M^2 = 1$.  This is indeed the case.  In appendix we give a proof of the corresponding statement for corner algebras.  The proofs for the two statements are similar, though it is a bit more difficult in the corner case since we cut off forms at $\frac{n-1}{2}$.  We are unable to prove directly from the definition that $M$ squares to the identity.  We note that $\Omega(\phi \otimes \psi \otimes \phi \otimes \psi) = \Omega^*(\phi \otimes \psi \otimes \phi \otimes \beta)$ so that $\Omega(\phi \otimes \psi \otimes \phi \otimes \beta)$ is always an eigenvector of $M$ with eigenvalue 1.  Some examples of these matrices, which we call \textit{cross symmetry matrices}, can be found in appendix \ref{csex}. \\

\noindent The matrix defined by (\ref{Mdef}) differs from a Krawtchouk matrix by only a factor of $\eps_{k+j}$ and both types of matrices share the fundamental property that they square to a multiple of the identity \cite{fein1}\cite{fein2}.

\subsection{Type $\mathbb{H}$ Subordinate Algebras}
\noindent Let us now turn to the case of type $\mathbb{H}$ subordinate algebras., which is very similar to the previous case.  Recall that the factorization of $T_\pm$ is determined by certain commutation relations, e.g. (\ref{comrelT+}).  The only difference from the type $\R$ case is that now the set of all elements which commute with $\Cpq$ is spanned by $\{1, I, J, K \}$ and the set of all elements that anti-commute with all vectors is spanned by $\{ \gamma, I \gamma, J \gamma, K \gamma \}$.  \\

\noindent For concreteness we will focus on $T_+$ with $\Gamma = \gamma_p$ and $p$ even.  Based on table \ref{T+table} we have that
\begin{equation}
T_+(\phi, \psi, \alpha, \beta) = \sum_{q_1, q_2 \in \{1, I, J, K \}} C_{q_1, q_2} (\phi, q_1 \alpha)_-(\psi, q_2 \beta)_-,\label{T+exp}
\end{equation}
where $C_{q_1,q_2} \in \R$.  From (\ref{sym2}) we get that $C_{q_1,q_2} = C_{q_2,q_1}$.  It follows from (\ref{conjsym}) and the definition of $T_+$ that
\begin{equation*}
T_+(I \phi, I \psi, \alpha, \beta) = T_+(J \phi, J \psi, \alpha, \beta) = T_+(K \phi, K \psi, \alpha, \beta) = T_+(\phi, \psi, \alpha, \beta).
\end{equation*}
Putting the identity $T_+(I\phi, I\psi, \alpha, \beta) = T_+(\phi, \psi, \alpha, \beta)$ into (\ref{T+exp}) and using the conjugate symmetry of $(\cdot, \cdot)_-$ gives
\begin{equation*}
\sum_{q_1, q_2 \in \{ 1, I, J, K \}} C_{q_1, q_2}(\phi, I q_1 \alpha)_- (\psi, I q_2 \beta)_- = \sum_{q_1, q_2 \in \{ 1, I, J, K \}} C_{q_1, q_2}(\phi, q_1 \alpha)_- (\psi, q_2 \beta)_-.
\end{equation*}
Equating the coefficients on $(\phi, J \alpha)_- (\psi, K \beta)_-$ shows that $-C_{J, K} = C_{K,J}$.  But since $C_{J,K} = C_{K,J}$, this means that $C_{J,K} = 0 = C_{K,J}$ and, by permuting $I, J,$ and $K$, we get that $0 = C_{I, J} =  C_{J, I} = C_{I, K} = C_{K,I}$.  Equating the coefficients on $(\phi, I \alpha)_- (\psi, J \beta)_-$ shows that $-C_{1, K} = C_{I, J} = 0$.  Finally, equating the coefficients on $(\phi, \alpha)_- (\psi, \beta)_-$ shows that $C_{I,I} = C_{1,1}$ and, by symmetry, $C_{J, J} = C_{1,1} = C_{K,K}$.  Thus (\ref{T+exp}) simplifies to
\begin{equation*}
T_+(\phi, \psi, \alpha, \beta) = C \sum_{q \in \{1, I, J, K \}} (\phi, q \alpha)_-(\psi, q \beta)_-.
\end{equation*}
The same technique used in the previous section determines $C$ give factorizations in the following tables.
\begin{table}[!h]
\begin{center}
\begin{tabular}{|c|c|c|}
\hline
\backslashbox[0.5cm]{$~~~~p$}{$\Gamma$} & $\gamma_p$ & $\gamma_q$ \\ \hline
even & $2^{(p+q-2)/2} \displaystyle\sum_{q \in \{1, I, J, K\}} (\phi, q \alpha)_- (\psi, q \beta)_- $ & $2^{(p+q-2)/2} \displaystyle\sum_{q \in \{1, I, J, K\}} (\phi, q \alpha)_- (\psi, q \beta)_-$ \\ \hline
odd & $2^{(p+q-2)/2} \displaystyle\sum_{q \in \{1, I, J, K\}} (\phi, q \alpha)_+ (\psi, q \beta)_+$ & $-2^{(p+q-2)/2} \displaystyle\sum_{q \in \{1, I, J, K\}} (\phi, q \alpha)_+ (\psi, q \beta)_+$ \\ \hline
\end{tabular}
\caption{Factorization of $T_+(\phi,\psi,\alpha,\beta)$ for type $\mathbb{H}$ subordinate algebras.} \label{HT+table}
\end{center}
\end{table}

\begin{table}[!h]
\begin{center}
\begin{tabular}{|c|c|c|}
\hline
\backslashbox[0.5cm]{$~~~~p$}{$\Gamma$} & $\gamma_p$ & $\gamma_q$ \\ \hline
even & $2^{(p+q-2)/2} \displaystyle\sum_{q \in \{1, I, J, K\}} (\phi, q \alpha)_+ (\psi, q \beta)_+$ & $2^{(p+q-2)/2} \displaystyle\sum_{q \in \{1, I, J, K\}} (\phi, q \alpha)_+ (\psi, q \beta)_+$ \\ \hline
odd & $-2^{(p+q-2)/2} \displaystyle\sum_{q \in \{1, I, J, K\}} (\phi, q \alpha)_- (\psi, q \beta)_-$ & $2^{(p+q-2)/2} \displaystyle\sum_{q \in \{1, I, J, K\}} (\phi, q \alpha)_- (\psi, q \beta)_-$ \\ \hline
\end{tabular}
\caption{Factorization of $T_-(\phi,\psi,\alpha,\beta)$ for type $\mathbb{H}$ subordinate algebras.} \label{HT-table}
\end{center}
\end{table}

\subsubsection{Cross-Symmetry}
Using the same procedure used in section (\ref{Rsubcs}), we define maps
\begin{align}
\Omega &: (\phi, \psi, \alpha, \beta) \mapsto \left(\begin{array}{c} \om_0(\phi, \psi) \cdot \om_0(\alpha, \beta) \\ \vdots \\ \om_{p+q}(\phi, \psi) \cdot \om_{p+q}(\alpha, \beta) \end{array}\right)\label{omdef} \\
\Omega^* &: (\phi, \psi, \alpha, \beta) \mapsto \left(\begin{array}{c} \displaystyle \sum_{q \in \{1, I, J, K\}} \om_0(\alpha, q \psi) \cdot \om_0(\phi, q \beta) \\ \vdots \\ \displaystyle \sum_{q \in \{1, I, J, K\}} \om_{p+q}(\alpha, q \psi) \cdot \om_{p+q}(\phi, q \beta) \end{array}\right). \label{om*def}
\end{align}
We have that
\begin{equation*}
\Omega^*(\phi, \psi, \alpha, \beta) = M \Omega(\phi, \psi, \alpha, \beta)
\end{equation*}
where $M$ is now the matrix whose $j+1, k+1$ component is
\begin{equation*}
2 \eps_{k+j} \sum_{m=0}^{\min\{j,k\}} (-1)^m {k \choose m} {p+q-k \choose j-m},
\end{equation*}
where $\eps_k$ is the same as before.  Note that this has an extra factor of 2 when compared with (\ref{Mdef}).  This is because the factorization of $T$ in a type $\mathbb{H}$ algebra carries a factor of $2^{(p+q-2)/2}$ whereas the factor is $2^{(p+q)/2}$ in a type $\R$ algebra.  Evidently, besides this factor of two, the cross-symmetry matrices for a type $\mathbb{H}$ algebra and a type $\R$ algebra of the same dimension are the same.  Therefore, we now have that $M^2 = 4\cdot1$.

\subsubsection{Example: (1,3)}
We will now show how we can use $T_\pm$ to derive specialized Fierz identities.  Because subordinate algebras have two bilinear forms and both $T_\pm$ are non-zero (when we get to corner algebras we will see that one of $T_\pm$ always vanishes), there are many identities.  We will therefore only derive a few to give a feel for how the process works.  Many of the familiar Fierz identities for $Cl_{1,3}$ are derivable solely from the factorizations of $T_\pm$.  From table \ref{HT+table} we have
\begin{equation}
\sum_{k=0}^4 \om_k^+(\phi, \psi) \cdot \om_k^+(\alpha, \beta) = 2 \sum_{q\in\{1,I,J,K\}} (\phi, q \alpha)_+ (\psi, q \beta)_+. \label{T+gq}
\end{equation}
Specializing to $\phi = \alpha = \beta = \psi$ and noticing that from $\ref{vpp+}$ only the 0, 1, and 4 forms do not vanish gives
\begin{equation*}
(\psi, \psi)_+^2 + \om_1^+(\psi, \psi)^2 + \om_4^+(\psi, \psi)^2 = 2(\psi, \psi)_+^2.
\end{equation*}
Note that $(\psi, I \psi)_+ = (\psi, J \psi)_+ = (\psi, K \psi)_+ = 0$ since  $I \gamma_p$ is anti-symmetric (since $(I \gamma_p)^2 = 1$).  Noticing that $\om_4^+(\psi, \psi)^2 = g(\gamma, \gamma) (\psi, \gamma \psi)_+^2 = -(\psi, \psi)_-^2$ gives the familiar Fierz identitiy \cite{lounesto}
\begin{equation}
\om_1^+(\psi, \psi)^2 = (\psi, \psi)_+^2 + (\psi, \psi)_-^2. \label{id1}
\end{equation}
If we now let $P$ be a pure unit quaternion and make the replacements $\phi = \psi$, $\alpha = \psi$, and $\beta = \gamma P \psi$ in (\ref{T+gq}) then $\om_k^+(\alpha, \beta)$ becomes $\om_k^-(\psi, P \psi)$.  From (\ref{vpp+}) and (\ref{vppH}) the only non-vanishing term on the left side is $\om_1^+(\psi, \psi) \cdot \om_1^-(\psi, P \psi)$.  Similarly, the right side simplifies to $(\psi, \psi)_+(\psi, P \psi)_- = 0$.  We therefore get the identitiy
\begin{equation}
\om_1^+(\psi, \psi) \cdot \om_1^-(\psi, P \psi) = 0. \label{id2}
\end{equation}
\noindent From tables \ref{HT+table} and \ref{HT-table} we see that
\begin{align*}
\sum_{k=0}^4 \om_k^-(\phi, \psi) \cdot \om_k^-(\alpha, \beta) &= -2 \sum_{q\in\{1,I,J,K\}} (\phi, q \alpha)_+ (\psi, q \beta)_+ \\
\sum_{k=0}^4 (-1)^k \om_k^-(\phi, \psi) \cdot \om_k^-(\alpha, \beta) &= 2 \sum_{q\in\{1,I,J,K\}} (\phi, q \alpha)_- (\psi, q \beta)_-.
\end{align*}
Subtracting them gives
\begin{equation*}
\om_1^-(\phi, \psi) \cdot \om_1^-(\alpha, \beta) + \om_3^-(\phi, \psi) \cdot \om_3^-(\alpha, \beta) = - \sum_{q\in\{1,I,J,K\}} \left( (\phi, q \alpha)_+ (\psi, q \beta)_+ + (\phi, q \alpha)_- (\psi, q \beta)_- \right).
\end{equation*}
Letting $P$ be a pure unit quaternion, as before, and specializing to $\psi = P \phi$, $\alpha = \phi$, $\beta = P \phi$, we get
\begin{align*}
\om_1^-(\phi, P \phi)^2 + \om_3^-(\phi, P \phi)^2 &= - \sum_{q\in\{1,I,J,K\}} \left( (\phi, q \phi)_+ (P \phi, q P \phi)_+ + (\phi, q \phi)_- (P \phi, q P \phi)_- \right) \\
&= - \sum_{q\in\{1,I,J,K\}} \left( (\phi, q \phi)_+^2 + (\phi, q \phi)_-^2 \right).
\end{align*}
By (\ref{vpp+}) and (\ref{vpp-}), the only non-zero term on the left side is $\om_1(\phi, P \phi)^2$ and the only non-zero terms on the right side are $(\phi, \phi)_+^2$ and $(\phi, \phi)_-^2$.  Substituting (\ref{id1}) gives
\begin{equation}
\om_1^-(\phi, P \phi)^2 = - \om_1^+(\phi,\phi)^2. \label{id3}
\end{equation}
The identities (\ref{id2}) and (\ref{id3}) are known In the Dirac treatment of $Cl_{1,3}$ but with $P$ replaced by $i$. \\

\noindent We will derive one more identity.  From tables \ref{HT+table} and \ref{HT-table} we have
\begin{align*}
\sum_{k=0}^4 \om_k^+(\phi,\psi) \cdot \om_k^+(\alpha, \beta) = 2 \sum_{q \in \{1,I,J,K\}} (\phi, q \alpha)_+ (\psi, q \beta)_+ \\
\sum_{k=0}^4 (-1)^k \om_k^+(\phi,\psi) \cdot \om_k^+(\alpha, \beta) = - 2 \sum_{q \in \{1,I,J,K\}} (\phi, q \alpha)_- (\psi, q \beta)_-.
\end{align*}
Letting $P$ be a pure unit quaternion, if we make the replacements $\psi \to P \phi, \beta \to P \phi$, and $\alpha \to \phi$, then by (\ref{vpp+}), (\ref{vppH}), and (\ref{conjsym}) adding these two equations yields
\begin{equation*}
\om_2^+(\phi, P \phi)^2 = (\phi, \phi)_+^2 - (\phi,\phi)_-^2.
\end{equation*}


\subsection{Type $\R$ Corner Algebras}
\noindent In a corner algebra $\gamma = \pm 1$, which means that elements in the representation of $\Cpq$ are self-dual, i.e. $e_I = \pm e_{I^c}$ (where $I^c$ is the complement of $I$).  In particular, this means that $\gamma_p = \pm \gamma_q$ so that there is really only one bilinear form, which we denote by $(\cdot, \cdot)$.  We define $\om_k$ to be $\om_k^+ = \om_k^-$.  If $q$ is even then $g(e_I, e_I) = g(e_{I^c}, e_{I^c})$.  Since $|I|$ and $|I^c|$ have different parities (since $p+q$ is always odd in corner algebras), this means that $T_-$ must vanish.  On the other hand, if $q$ is odd then $g(e_I e_I) = - g(e_{I^c}, e_{I^c})$ so that $T_+$ vanishes.  That one of these tensors must vanish is consistent with our results in the previous section since in the subordinate algebra case we saw that the factorization of one of $T_\pm$ always relied on the existence of an element that anti-commutes with $\Rpq$.  There is no such element in a corner algebra since $\gamma = \pm 1$.  Analogous arguments as those used in the previous section give
\begin{equation*}
T_+(\phi, \psi, \alpha, \beta) = \begin{cases}
2^{(p+q-1)/2} (\phi, \alpha) (\psi, \beta) &\text{ if $q$ is even} \\
0 &\text{ if $q$ is odd}
\end{cases}
\end{equation*}
\begin{equation*}
T_-(\phi, \psi, \alpha, \beta) = \begin{cases}
0 &\text{ if $q$ is even} \\
2^{(p+q-1)/2} (\phi, \alpha) (\psi, \beta) &\text{ if $q$ is odd}.
\end{cases}
\end{equation*}
When working in a corner algebra, we will sometimes write $T$ for whichever of $T_\pm$ is non-zero.  We can write $T$ for a general corner algebra as
\begin{align}
T(\phi,\psi,\alpha,\beta) &= \sum_{|I| \le \frac{p+q-1}{2}} (-1)^{q|I|} g(e_I, e_I) (\phi, e_I \psi) (\alpha, e_I \beta) \nonumber \\
&= \sum_{k=0}^{\frac{p+q-1}{2}} (-1)^{qk} \om_k(\phi, \psi) \cdot \om_k(\alpha, \beta). \label{Tom}
\end{align}

\subsubsection{Cross-Symmetry}
The calculations from section \ref{Rsubcs} carry over to the type $\R$ corner algebra case with the exception that now we cut forms off at $\frac{p+q-1}{2}$.  We thus define maps
\begin{align}
\Omega &: (\phi, \psi, \alpha, \beta) \mapsto \left(\begin{array}{c} \om_0(\phi, \psi) \cdot \om_0(\alpha, \beta) \\ \vdots \\ \om_{p+q}(\phi, \psi) \cdot \om_{p+q}(\alpha, \beta) \end{array}\right)\label{omdef} \\
\Omega^* &: (\phi, \psi, \alpha, \beta) \mapsto \left(\begin{array}{c} \om_0(\alpha, \psi) \cdot \om_0(\phi, \beta) \\ \vdots \\ \om_{\frac{p+q-1}{2}}(\alpha, \psi) \cdot \om_{\frac{p+q-1}{2}}(\phi, \beta) \end{array}\right). \label{om*def}
\end{align}
We have that
\begin{equation*}
\Omega^*(\phi, \psi, \alpha, \beta) = M \Omega(\phi, \psi, \alpha, \beta)
\end{equation*}
where $M$ is now the matrix whose $j+1, k+1$ component is
\begin{equation*}
\frac{1}{2^{(p+q-1)/2}} (-1)^{\frac{1}{2}(k+j)(k+j-1)} \sum_{m=0}^{\min\{j,k\}} (-1)^m {k \choose m} {p+q-k \choose j-m}.
\end{equation*}
As in the type $\R$ subordinate case, $M$ squares to the identity.  Since forms get cut off at $\frac{p+q-1}{2}$, the cross-symmetry matrices are less symmetrical than those for subordinate algebras, as seen in appendix \ref{csex}.

\subsubsection{Example: (10, 1)}
While the cross-symmetry matrices do not depend much on the signature, Fierz identities involving fewer than four spinors do.  Therefore it is best to illustrate with an example.  We do this with the signature used in M-theory: (10,1).  The cross symmetry relation in (10,1) is 
\begin{equation}
32 \left(\begin{array}{c} \om_0(\alpha, \psi) \cdot \om_0(\phi, \beta) \\ \om_1(\alpha, \psi) \cdot \om_1(\phi, \beta) \\ \om_2(\alpha, \psi) \cdot \om_2(\phi, \beta) \\ \om_3(\alpha, \psi) \cdot \om_3(\phi, \beta) \\ \om_4(\alpha, \psi) \cdot \om_4(\phi, \beta) \\ \om_5(\alpha, \psi) \cdot \om_5(\phi, \beta) \end{array}\right) =  \left( \begin {array}{cccccc} 1&1&-1&-1&1&1\\\noalign{\medskip}11&-9&
-7&5&3&-1\\\noalign{\medskip}-55&-35&19&7&1&5\\\noalign{\medskip}-165&
75&21&5&11&-5\\\noalign{\medskip}330&90&6&22&-6&10\\\noalign{\medskip}
462&-42&42&-14&14&-10\end {array}
 \right) \left(\begin{array}{c} \om_0(\phi, \psi) \cdot \om_0(\alpha, \beta) \\  \om_1(\phi, \psi) \cdot \om_1(\alpha, \beta) \\  \om_2(\phi, \psi) \cdot \om_2(\alpha, \beta) \\  \om_3(\phi, \psi) \cdot \om_3(\alpha, \beta) \\  \om_4(\phi, \psi) \cdot \om_4(\alpha, \beta) \\ \om_5(\phi, \psi) \cdot \om_5(\alpha, \beta) \end{array}\right) \label{101basic}
\end{equation}
Specializing to $\alpha = \phi$ makes the two vectors equal.  Subtracting the left vector from both sides then gives us
\begin{equation*}
\left( \begin {array}{cccccc} 
-31 &1&-1&-1&1&1\\\noalign{\medskip}
11&-41&-7&5&3&-1\\\noalign{\medskip}
-55&-35&-13&7&1&5\\\noalign{\medskip}
-165&75&21&-27&11&-5\\\noalign{\medskip}
330&90&6&22&-38&10\\\noalign{\medskip}
462&-42&42&-14&14&-42\end {array}
 \right) \left(\begin{array}{c} \om_0(\phi, \psi) \cdot \om_0(\phi, \beta) \\  \om_1(\phi, \psi) \cdot \om_1(\phi, \beta) \\  \om_2(\phi, \psi) \cdot \om_2(\phi, \beta) \\  \om_3(\phi, \psi) \cdot \om_3(\phi, \beta) \\  \om_4(\phi, \psi) \cdot \om_4(\phi, \beta) \\ \om_5(\phi, \psi) \cdot \om_5(\phi, \beta) \end{array}\right) = 0.
\end{equation*}
The above matrix has rank three.  Reducing it gives these three linearly independent identities
\begin{equation*}
\left( \begin {array}{cccccc}
15 & 0 & 0 & 1 & -1 & 0 \\
0 & 30 & 0 & 2 & -7 & 5 \\
0 & 0 & 6 & -6 & 5 & -5 \\
\end{array} \right) 
\left(\begin{array}{c} \om_0(\phi, \psi) \cdot \om_0(\phi, \beta) \\  \om_1(\phi, \psi) \cdot \om_1(\phi, \beta) \\  \om_2(\phi, \psi) \cdot \om_2(\phi, \beta) \\  \om_3(\phi, \psi) \cdot \om_3(\phi, \beta) \\  \om_4(\phi, \psi) \cdot \om_4(\phi, \beta) \\ \om_5(\phi, \psi) \cdot \om_5(\phi, \beta) \end{array}\right) = 0.
\end{equation*}
If we specialize even further to $\phi = \psi$ and $\beta = \psi$ then we get
\begin{align*}
30 \om_1(\psi, \psi) \cdot \om_1(\psi, \psi) + 5 \om_5(\psi, \psi) \cdot \om_5 (\psi, \psi) = 0 \\
6 \om_2(\psi, \psi) \cdot \om_1(\psi, \psi) - 5 \om_5(\psi, \psi) \cdot \om_5 (\psi, \psi) = 0. 
\end{align*}
Putting $\phi = \psi$ and $\beta = \alpha$ in (\ref{101basic}) and noticing from (\ref{vpp-}) that only $\om_k(\psi, \psi) \ne 0$ for $k = 1, 2, 5$, we get
\begin{equation*}
32 \left(\begin{array}{c} \om_0(\alpha, \psi) \cdot \om_0(\psi, \alpha) \\ \om_1(\alpha, \psi) \cdot \om_1(\psi, \alpha) \\ \om_2(\alpha, \psi) \cdot \om_2(\psi, \alpha) \\ \om_3(\alpha, \psi) \cdot \om_3(\psi, \alpha) \\ \om_4(\alpha, \psi) \cdot \om_4(\psi, \alpha) \\ \om_5(\alpha, \psi) \cdot \om_5(\psi, \alpha) \end{array}\right) = 
\left( \begin {array}{ccc}
1 & -1 & 1 \\
-9 & -7 & -1 \\
-35 & 19 & 5 \\
75 & 21 & -5 \\
90 & 6 & 10 \\
-42 & 42 & -10 \\
\end{array} \right) 
\left(\begin{array}{c} \om_1(\psi, \psi) \cdot \om_1(\alpha, \alpha) \\  \om_2(\psi, \psi) \cdot \om_2(\alpha, \alpha) \\  \om_5(\psi, \psi) \cdot \om_5(\alpha, \alpha) \end{array}\right). 
\end{equation*}
We can use (\ref{om-sym}) to write $\om_k(\alpha, \psi)$ in terms of $\om_k(\psi, \alpha)$.  Putting $\om_k(\psi,\alpha)^2$ for $\om_k(\psi, \alpha) \cdot \om_k(\psi, \alpha)$ and absorbing into the matrix any minus sign that may result from going from $\om_k(\alpha, \psi)$ to $\om_k(\psi, \alpha)$ gives
\begin{equation*}
32 \left(\begin{array}{c} \om_0(\psi, \alpha)^2 \\ \om_1(\psi, \alpha)^2 \\ \om_2(\psi, \alpha)^2 \\ \om_3(\psi, \alpha)^2 \\ \om_4(\psi, \alpha)^2 \\ \om_5(\psi, \alpha)^2 \end{array}\right) = 
\left( \begin {array}{ccc}
-1 & 1 & -1 \\
-9 & -7 & -1 \\
-35 & 19 & 5 \\
-75 & -21 & 5 \\
-90 & -6 & -10 \\
-42 & 42 & -10 \\
\end{array} \right) 
\left(\begin{array}{c} \om_1(\psi, \psi) \cdot \om_1(\alpha, \alpha) \\  \om_2(\psi, \psi) \cdot \om_2(\alpha, \alpha) \\  \om_5(\psi, \psi) \cdot \om_5(\alpha, \alpha) \end{array}\right). 
\end{equation*}

\subsection{Type $\mathbb{H}$ Corner Algebras}
Mimicking what was done in the type $\R$ corner algebra case, we see that for a type $\mathbb{H}$ corner algebra, one of $T_\pm$ also vanishes.  Letting $T$ be the non-zero one, we can write
\begin{equation*}
T(\phi,\psi,\alpha,\beta) = \sum_{k=0}^{\frac{p+q-3}{2}} (-1)^{qk} \om_k(\phi, \psi) \cdot \om_k(\alpha, \beta).
\end{equation*}
Its factorization is
\begin{equation*}
T(\phi, \psi, \alpha, \beta) = 2^{(p+q-3)/2} \sum_{q \in \{1, I, J, K\}} (\phi, q \alpha) (\psi, q \beta).
\end{equation*}

\subsubsection{Cross-Symmetry}
As in the subordinate case, the type $\mathbb{H}$ corner algebras are similar to the type $\R$ corner algebras.  Since we still cut forms off at $\frac{p+q-1}{2}$,  we define maps
\begin{align}
\Omega &: (\phi, \psi, \alpha, \beta) \mapsto \left(\begin{array}{c} \om_0(\phi, \psi) \cdot \om_0(\alpha, \beta) \\ \vdots \\ \om_{p+q}(\phi, \psi) \cdot \om_{p+q}(\alpha, \beta) \end{array}\right)\label{omdef} \\
\Omega^* &: (\phi, \psi, \alpha, \beta) \mapsto \left(\begin{array}{c} \sum_{q \in \{1,I,J,K\}}\om_0(\alpha, q \psi) \cdot \om_0(\phi, q \beta) \\ \vdots \\ \sum_{q \in \{1,I,J,K\}} \om_{\frac{p+q-1}{2}}(\alpha, q \psi) \cdot \om_{\frac{p+q-1}{2}}(\phi, q \beta) \end{array}\right). \label{om*def}
\end{align}
We have that
\begin{equation*}
\Omega^*(\phi, \psi, \alpha, \beta) = M \Omega(\phi, \psi, \alpha, \beta)
\end{equation*}
where $M$ is the matrix whose $j+1, k+1$ component is
\begin{equation*}
\frac{1}{2^{(p+q-3)/2}} (-1)^{\frac{1}{2}k(k-1)} \sum_{m=0}^{\min\{j,k\}} (-1)^m {k \choose m} {p+q-k \choose j-m}.
\end{equation*}
We see that $M$ is two times the cross-symmetry matrix for a type $\R$ corner algebra of the same dimension.  Therefore, like in the type $\mathbb{H}$ subordinate case, $M$ squares to four times the identity.

\section{Conclusion}
We have showed the existence of a class of Fierz identities for real representations.  Along the way we have proved the existence of an entire class of involutory real matrices for any dimension.  There are still issues that we would like to explore in future work.  These include
\begin{itemize}
\item A better way to look at real representations of type $\C$ algebras to get similar identities.
\item Looking for a similar construction with complex and/or quaternionic representations of Clifford algebras.
\item Seeing if these involutory matrices have any deep significance and exploring their relationship to Krawtchouk matrices (which show up in a variety of places \cite{fein1}).
\item The relationship between the constructions in this paper with the two-spinor calculus of Penrose, Rindler, and Newman \cite{penroserindler}, \cite{penrosenewman}.
\item If the similarities between $\R$ and $\mathbb{H}$ in these constructions are manifestation of a deeper duality between $\R$ and $\mathbb{H}$.
\end{itemize}
Addressing the second point, the reason we cannot simply repeat this construction for complex and quaternionic representations is that now the condition (\ref{tis-1}) is not the same as $e_I^t = e_I^{-1}$ (as it is for real representations).  This, in turn, gives different symmetries in $T_\pm$ so that $T_\pm$ will not necessarily factor as $\gamma_{p(q)} \otimes \gamma_{p(q)}$.

\appendix
\appendixpage
\section{Proof that $M$ is surjective}
We will prove this by induction, the inductive step being
\begin{thm}
If the map $\Omega$ is surjective for $(p,q)$ and $(8,0)$ spinors, then it is surjective for $(p+8,q)$ spinors.
\end{thm}
\begin{proof}
Let $\{ f_i \}$ and $\{ e_i \}$ be standard generators for $\Cpq$ and $Cl_{8,0}$, respectively.  Then the set $\{ e_i \otimes 1, \gamma \otimes f_j \}$ generates $Cl_{p+8,q}$, where $\gamma = \prod_{i=1}^8 e_i$.  Note that this is a generating set since $\gamma^2 = 1$ and $\gamma$ anti-commutes with each $e_i$.  Let $\{ \phi_i \}$ and $\{ \psi_j \}$ be bases for $\Cpq$ and $Cl_{8,0}$ spinors, respectively, so that $\{ \psi_j \otimes \phi_i \}$ is a basis for $Cl_{p+8,q}$ spinors.  The $\gamma_p$ in $Cl_{p+8,q}$ is equal to $\gamma \otimes \gamma_p$ (abusing notation so that this $\gamma_p$ is the $\gamma_p$ of $\Cpq$) if $p$ is even and $(1 \otimes \gamma_p)$ is $p$ is odd.  We will assume that $p$ is even (the $p$ odd case is analogous).  We will denote the inner product on $(p+8, q)$ spinors as $\langle \cdot, \cdot \rangle$: 
\begin{align*}
\langle \psi_1 \otimes \phi_1, \psi_2 \otimes \phi_2 \rangle &= (\psi_1^t \otimes \phi_1^t)(\gamma \otimes \gamma_p) (\psi_2 \otimes \phi_2) \\
&= (\psi_1, \psi_2)_+ (\phi_1, \phi_2), 
\end{align*}
where $(\psi_1, \psi_2)_+ = \psi_1^t \gamma \psi_2$.  \\

\noindent Consider now $\om_k(\psi_1 \otimes \phi_1, \psi_2 \otimes \phi_2)$.  The coefficient on 
\begin{equation}
(e_{i_1} \otimes 1) \ldots (e_{i_m} \otimes 1) (\gamma \otimes f_{j_1}) \ldots (\gamma \otimes f_{j_{k-m}}) \label{ext}
\end{equation}
is 
\begin{equation}
\langle \psi_1 \otimes \phi_1, (e_{i_1} \otimes 1) \ldots (e_{i_m} \otimes 1) (\gamma \otimes f_{j_1}) \ldots (\gamma \otimes f_{j_{k-m}}) \psi_2 \otimes \phi_2  \rangle. \label{coeff}
\end{equation}
If $k - m$ is even then all of the $\gamma$'s with the $f_{j_i}$'s cancel each other so that $\ref{ext}$ becomes
\begin{equation*}
(e_{i_1} \ldots e_{i_m}) \otimes (f_{j_1} \ldots f_{j_{k-m}})
\end{equation*}
and $\ref{coeff}$ becomes
\begin{gather*}
\langle \psi_1 \otimes \phi_1, (e_{i_1} \ldots e_{i_m}) \otimes (f_{j_1} \ldots f_{j_{k-m}}) \psi_2 \otimes \phi_2 \rangle \\
= (\psi_1, e_{i_1} \ldots e_{i_m} \psi_2)_+ (\phi_1, f_{j_1} \ldots f_{j_{k-m}} \phi_2).
\end{gather*}
On the other hand, if $k - m$ is odd then $\ref{ext}$ becomes
\begin{equation*}
(e_{i_1} \ldots e_{i_m} \gamma) \otimes (f_{j_1} \ldots f_{j_{k-m}})
\end{equation*}
and $\ref{coeff}$ becomes 
\begin{equation*}
(\psi_1, e_{i_1} \ldots e_{i_m} \gamma \psi_2)_+ (\phi_1, f_{j_1} \ldots f_{j_{k-m}} \phi_2).
\end{equation*}
Using the fact that, because of duality,
\begin{equation*}
\sum_{|I| = k} (\psi_1, e_I \gamma \psi_2) e_I \gamma = \sum_{|I| = 8-k} (\psi_1, e_I \psi_2) e_I
\end{equation*}
we therefore have that (up to some signs),
\begin{align*}
\om_k(\psi_1 \otimes &\phi_1, \psi_2 \otimes \phi_2) = \\
&\sum_{|I| = k}  (\psi_1, e_I \psi_2)_+ (\phi_1, \phi_2) e_I \otimes 1 + \sum_{|I| = k-1, |J| = 1}  (\psi_1, e_I \gamma \psi_2)_+ (\phi_1, f_J \phi_2) [(e_I \gamma) \otimes f_J] + \ldots \\
&= \om_k(\psi_1, \psi_2) \otimes \om_0(\phi_1, \phi_2) + \om_{8-(k-1)}(\psi_1, \psi_2) \otimes \om_1(\phi_1, \phi_2) \\
&+ \om_{k-2}(\psi_1, \psi_2) \otimes \om_{2}(\phi_1, \phi_2) + \om_{8 - (k-3)}(\psi_1, \psi_2) \otimes \om_3(\phi_1, \phi_2) + \ldots
\end{align*} 
where it is understood that $\om_k(\psi_1, \psi_2) = 0$ if $k > 8$ and $\om_m(\phi_1, \phi_2) = 0$ if $m > n$.  Note that if $m > \frac{n-1}{2}$, then $\om_m(\phi_1, \phi_2)$ is really, up to sign, $\om_{n-m}(\phi_1, \phi_2)$ (since in a corner algebra $e_1 \ldots e_{n} = \pm 1$ so we cut forms off at $\frac{n-1}{2}$ forms).  We have
\begin{align*}
\om_k(&\psi_1 \otimes \phi_1, \psi_2 \otimes \phi_2) \cdot \om_k(\psi_3 \otimes \phi_3, \psi_4 \otimes \phi_4) \\
&= \om_k(\psi_1, \psi_2) \cdot \om_k(\psi_3, \psi_4) \om_0(\phi_1, \phi_2) \cdot \om_0(\phi_3, \phi_4) \\
&+ \om_{9 - k}(\psi_1, \psi_2) \cdot \om_{9-k}(\psi_3, \psi_4) \om_1(\phi_1, \phi_2) \cdot \om_1(\phi_3, \phi_4) \\
&+ \om_{k-2}(\psi_1, \psi_2) \cdot \om_{k-2}(\psi_3, \psi_4) \om_2(\phi_1, \phi_2) \cdot \om_2(\phi_3, \phi_4) \\
&+ \ldots \\
&= \sum_{u} \om_{\sigma_k(u)}(\psi_1, \psi_2) \cdot \om_{\sigma_k(u)}(\psi_3, \psi_4) \om_u(\phi_1, \phi_2) \cdot \om_u(\phi_3, \phi_4)
\end{align*}
where
\begin{equation*}
\sigma_k(u) = \begin{cases}
k - u &\text{ if $u$ is even} \\
8 - (k -u) &\text{ if $u$ is odd}.
\end{cases}
\end{equation*}
We will now show that the map $\Omega$ is surjective from the four-fold tensor product of $(p+8,q)$ spinors, $\bigotimes_{i=1}^4 S$, onto $\R^\frac{n+9}{2}$.  To do this, we will show that for any $k$ between $0$ and $\frac{n+7}{2}$, we can always find an element of $A \in \bigotimes_{i=1}^4 S$ such that the $k$-form component of  $\Omega(A)$ is $1$ and all other components are $0$. \\

\noindent We will first show that a term of the form $\om_j(\psi_1, \psi_2) \cdot \om_j(\psi_3, \psi_4) \om_i(\phi_1, \phi_2) \cdot \om_i(\phi_3,\phi_4)$ shows up only once in $\Omega([\psi_1 \otimes \phi_1] \otimes [\psi_2 \otimes \phi_2] \otimes [\psi_3 \otimes \phi_3] \otimes [\psi_4 \otimes \phi_4])$.  Consider first the case that $i$ is even.  Then the term either shows up as part of a $i + j$-form or it could have come from $\Cpq$ duality in which case the $\om_i$ part should be thought of as $\om_{n-i}$.  In the latter case, $n-i$ must be odd since $n$ is odd and $i$ is even.  Therefore there were really $8-j$ $e_i's$ (since when there is an odd amount of $f_i$'s there is an extra $\gamma$) so that this is a term of a $n-i + 8-j$ form.  In the first case we must have that $i+j \le \frac{n+7}{2}$ and in the second we would have that $n - i + 8 - j \le \frac{n+7}{2}$.  However, these two inequalities are incompatible since the second one implies that $i+j \ge \frac{n+9}{2}$. \\

\noindent If $i$ is odd then, by similar reasoning, the two cases are that $i+8-j \le \frac{n+7}{2}$ or $n - i +j \le \frac{n+7}{2}$.  However, these are also incompatible since the first implies that $i - j \le \frac{n-9}{2}$ and the second implies that $i - j \ge \frac{n-7}{2}$. \\

\noindent Now let $k$ be some number between 0 and $\frac{n+7}{2}$ and let $i$ be such that $0 \le i \le \min\{k, n \}$.  By assumption, the $\Omega$ maps for $(p,q)$ and $(8,0)$ spinors are both 
surjective.  Therefore we can find $(p,q)$ spinors $\{ \phi_u^v\}$ and $(8,0)$ spinors $\{ \psi_r^s \}$ such that
\begin{equation*}
\sum_{\alpha} c_\alpha \om_m(\phi_1^\alpha, \phi_2^\alpha) \cdot \om_m(\phi_3^\alpha, \phi_4^\alpha) = \begin{cases}
1 &\text{ if $m = i$} \\
0 &\text{ otherwise}
\end{cases}
\end{equation*}
and
\begin{equation*}
\sum_{\beta} d_\beta \om_m(\psi_1^\beta, \phi_2^\beta) \cdot \om_m(\psi_3^\beta, \psi_4^\beta) = \begin{cases}
1 &\text{ if $m = \sigma_k(i)$} \\
0 &\text{ otherwise},
\end{cases}
\end{equation*}
where $c_\alpha, d_\beta \in \R$.  We then have that
\begin{align*}
\sum_{\alpha,\beta} &c_\alpha d_\beta \om_m(\psi_1^\beta \otimes \phi_1^\alpha, \psi_2^\beta \otimes \phi_2^\alpha) \cdot \om_m(\psi_1^\beta \otimes \phi_1^\alpha, \psi_2^\beta \otimes \phi_2^\alpha) \\
&= \sum_\alpha c_\alpha \sum_\beta d_\beta \om_m(\psi_1^\beta \otimes \phi_1^\alpha, \psi_2^\beta \otimes \phi_2^\alpha) \cdot \om_m(\psi_1^\beta \otimes \phi_1^\alpha, \psi_2^\beta \otimes \phi_2^\alpha) \\
&= \sum_\alpha c_\alpha \sum_\beta d_\beta \sum_u \om_{\sigma_m(u)}(\psi_1^\beta, \psi_2^\beta) \cdot \om_{\sigma_m(u)}(\psi_3^\beta, \psi_4^\beta) \om_u(\phi_1^\alpha, \phi_2^\alpha) \cdot \om_u(\phi_3^\alpha, \phi_4^\alpha) \\
&= \sum_u \sum_\alpha c_\alpha \om_u(\phi_1^\alpha, \phi_2^\alpha) \cdot \om_u(\phi_3^\alpha, \phi_4^\alpha) \sum_\beta d_\beta \om_{\sigma_m(u)}(\psi_1^\beta, \psi_2^\beta) \cdot \om_{\sigma_m(u)}(\psi_3^\beta, \psi_4^\beta) \\
&= \begin{cases}
1 &\text{ if at some point $u=i$ and $\sigma_m(u) = \sigma_k(u)$} \\
0 &\text{ otherwise}
\end{cases} \\
&= \begin{cases}
1 &\text{ if $m = k$ } \\
0 &\text{ otherwise},
\end{cases}
\end{align*}
the last inequality coming from the fact that, as shown above, a term of the form $\om_{\sigma_k(i)}(\psi_1, \psi_2) \cdot \om_{\sigma_k(i)}(\psi_3, \psi_4) \om_i(\phi_1, \phi_2) \cdot \om_i(\phi_3,\phi_4)$ only appears in the component of $\Omega([\psi_1 \otimes \phi_1] \otimes [\psi_2 \otimes \phi_2] \otimes [\psi_3 \otimes \phi_3] \otimes [\psi_4 \otimes \phi_4])$ that is associated with $k$-forms (i.e. the $k+1$ component).
\end{proof}
\noindent We have thus reduced showing that $\Omega$ is surjective (and therefore that $M$ is an involutory matrix) for any corner algebra to checking it for signature (8,0) and the corner signatures in dimensions less than 8: (1,0),(2,1),(3,2), (4,3), and (0,7).  The map $\Omega$ is indeed surjective in these cases, as we have verified computationally in Maple.

\section{A Direct Proof that $M^2 = 1$}
We will first tackle the case when $\Cpq$ is a subordinate algebra.  Letting $N = p+q$, what we want to show is that the $(N+1) \times (N+1)$ matrix $M$ whose $j+1, k+1$ component is
\begin{equation}
\frac{1}{2^{N/2}} (-1)^{\frac{1}{2}(k+j)(k+j-1)} \sum_{m=0}^{\min \{j, k\}} (-1)^m {k \choose m} {N - k \choose j - m} \label{defM}
\end{equation}
squares to the identity.  We have that the $j+1, l+1$ component of $M^2$ is
\begin{align}
\frac{1}{2^N} &\sum_{k,m,n} (-1)^{\frac{1}{2}((k+j)(k+j-1) + (k+l)(k+l-1))} (-1)^{m+n} {k \choose m} {N - k \choose j - m} {l \choose m} {N - l \choose k - n} \notag \\
&= \hspace{-3pt} \frac{(-1)^{\frac{1}{2}(j(j-1) + l(l-1))}}{2^N} \hspace{-6pt} \sum_{k,m,n} (-1)^{k(j+l) + m + n} {k \choose m} \hspace{-2pt} {N - k \choose j - m} \hspace{-2pt} {l \choose m} \hspace{-2pt} {N - l \choose k - n}.\label{M2}
\end{align}
We can allow $k, m$ and $n$ to run from 0 to $\infty$ since any term that has $k, m$ or $n$ outside of  its defined limit will vanish.  We will make use of the following consequence of the residue theorem:
\begin{equation*}
\{ \text{coefficient of $z^n$ in the expansion of $f(z)$ centered at 0} \} = \frac{1}{2 \pi i} \oint_C \frac{f(z)}{z^{n+1}} dz,
\end{equation*}
where $C$ is a closed curve around $C$ and $f$ has no singularities on or inside $C$.  \\

\noindent Momentarily ignoring the factor outside of the summation in (\ref{M2}), we have
\begin{align}
&\sum_{k,m,n} (-1)^{k(j+l) + m + n} {k \choose m} {N - k \choose j - m} {l \choose n} {N - l \choose k - n} \label{1} \\
&= \hspace{-3pt} \left(\frac{1}{2\pi i}\right)^2 \hspace{-5pt} \sum_{k,m,n} (-1)^{k(j+l)+m+n} \hspace{-1pt} {k \choose m} \hspace{-3pt} \oint_{C_1} \hspace{-6pt} \frac{(1+z)^{N-k}}{z^{j-m+1}} dz {l \choose n} \hspace{-3pt} \oint_{C_2} \hspace{-6pt} \frac{(1+w)^{N-l}}{w^{k-n+1}} dw \label{2} \\
&= \hspace{-3pt} \left(\frac{1}{2\pi i}\right)^2 \hspace{-2pt} \sum_k (-1)^{k(j+l)} \hspace{-3pt} \oint_{C_1} \hspace{-2pt} \frac{(1+z)^{N-k}}{z^{j+1}} \sum_m {k \choose m} (-1)^m z^m  dz \hspace{-2pt} \oint_{C_2} \hspace{-2pt} \frac{(1+w)^{N-l}}{w^{k+1}} \sum_n \hspace{-2pt}{l \choose n}(-1)^n w^n dw \nonumber  \\
&= \left(\frac{1}{2\pi i}\right)^2 \sum_k (-1)^{k(j+l)} \oint_{C_1} \frac{(1+z)^{N-k} (1-z)^k}{z^{j+1}} dz \oint_{C_2} \frac{(1+w)^{N-l} (1-w)^l}{w^{k+1}} dw \nonumber \\
&=\left(\frac{1}{2\pi i}\right)^2  \hspace{-3pt} \oint_{C_1} \hspace{-3pt} \oint_{C_2} \hspace{-8pt} \frac{(1+z)^N (1+w)^{N-l}(1-w)^l}{z^{j+1}w} \sum_{k=0}^\infty \left( \frac{(-1)^{j+l}(1-z)}{(1+z) w} \right)^k \hspace{-6pt} dw dz \label{5} \\
&=\left(\frac{1}{2\pi i}\right)^2  \oint_{C_1} \oint_{C_2} \frac{(1+z)^N (1+w)^{N-l}(1-w)^l}{z^{j+1}w} \frac{1}{1- \frac{(-1)^{j+l}(1-z)}{(1+z) w}} dw dz \label{6} \\
&= \left(\frac{1}{2\pi i}\right)^2 \oint_{C_1} \oint_{C_2} \frac{(1+z)^{N+1} (1+w)^{N-l}(1-w)^l}{z^{j+1}((1+z)w - (-1)^{j+l} (1-z))} dw dz \label{7} \\
&= \frac{1}{2 \pi i} \oint_{C_1} \frac{(1+z)^N}{z^{j+1}} \left( \frac{1}{2 \pi i} \oint_{C_2} \frac{(1+w)^{N-l}(1-w)^l}{w-(-1)^{j+l}\frac{1-z}{1+z}} dw \right) dz \label{8} \\
&= \frac{1}{2 \pi i} \oint_{C_1} \frac{(1+z)^N}{z^{j+1}} \left(1+(-1)^{j+l} \frac{1-z}{1+z}\right)^{N-l}\left(1-(-1)^{j+l} \frac{1-z}{1+z}\right)^l dz  \label{9} \\
&= \frac{1}{2 \pi i} \oint_{C_1} \frac{\left(1 + z + (-1)^{j+l}(1-z)\right)^{N-l} \left(1+z - (-1)^{j+l}(1-z)\right)^l}{z^{j+1}} \label{10} dz \\
&= \frac{1}{2 \pi i} \oint_{C_1} \begin{cases}
\frac{(2z)^{N-l} 2^l}{z^{j+1}} dz &\text{ if $j+l$ is odd} \\
\frac{2^{N-l} (2z)^l}{z^{j+1}} dz &\text{ if $j+l$ is even}
\end{cases} \label{11} \\
&= 2^N \frac{1}{2 \pi i} \oint_{C_1} \begin{cases}
\frac{1}{z^{j+l+1-N}} dz &\text{ if $j+l$ is odd} \\
\frac{1}{z^{j-l+1}} dz &\text{ if $j+l$ is even}
\end{cases}. \label{12}
\end{align}
In going from (\ref{1}) to (\ref{2}) we need both $C_1$ and $C_2$ to circle the origin, in going from (\ref{5}) to (\ref{6}) we need $\frac{|1-z|}{|1+z||w|} < 1$, and in going from (\ref{8}) to (\ref{9}) we need $C_2$ to circle $(-1)^{j+l} \frac{1-z}{1+z}$.  This can easily be achieved if $C_1$ and $C_2$ are circles about the origin with the radius of $C_1$ be sufficiently small and the radius of $C_2$ being sufficiently big ($C_1: |z| = \frac{1}{4}$ and $C_2: |w| = 2$ will do).  Since $j+l$ cannot be both odd and equal to $N$ (since in a subordinate algebra $N$ is even), we have that (\ref{12}) vanishes if $j+l$ is odd.  If $j + l$ is even then it can only be non-zero if $j = l$, in which case it is equal to $2^N$.  Putting this into (\ref{M2}) establishes that $M^2 = 1$. \\

\noindent In a corner algebra, we cannot let $k$ go to infinity since we must cut off forms at degree $\frac{N-1}{2}$.  This was crucial in (\ref{5}) since it enabled us to only get a simple pole in the $w$ integral.  However, we have that
\begin{align*}
&\sum_{k=0}^{(N-1)/2} \sum_{m=0}^{\min\{j, k\}} \sum_{n=0}^{\min\{k,l\}} (-1)^{k(j+l)+m+n} {k \choose m} {N - k \choose j - m} {l \choose n} {N-l \choose k - n} \\
&= \sum_{k=\frac{N-1}{2}}^N \sum_{m=0}^j \sum_{n=0}^l (-1)^{(N-k)(j+l)+m+n} {N - k \choose m}{k \choose j - m} {l \choose n} {N-l \choose N - k - n} \\
&= (-1)^{N(j+l)} \hspace{-8pt} \sum_{k=\frac{N-1}{2}}^N \sum_{m=0}^j \sum_{n=0}^l (-1)^{-k(j+l)+j-m+l-n} {N - k \choose j-m}{k \choose m}{l \choose l - n}{N-l \choose N - k - l + n} \\
&= (-1)^{N(j+l)+j+l} \sum_{k=\frac{N-1}{2}}^N \sum_{m=0}^j \sum_{n=0}^l (-1)^{k(j+l)+m+n} {k \choose m} {N - k \choose j - m} {l \choose n} {N-l \choose N - l -(k - n)} \\
&= (-1)^{(N+1)(j+l)} \sum_{k=\frac{N-1}{2}}^N \sum_{m=0}^j \sum_{n=0}^l (-1)^{k(j+l)+m+n} {k \choose m} {N - k \choose j - m} {l \choose n} {N-l \choose k - n} \\
&= \sum_{k=\frac{N-1}{2}}^N \sum_{m=0}^{\min\{j,k\}} \sum_{n=0}^{\min\{k,l\}} (-1)^{k(j+l)+m+n} {k \choose m} {N - k \choose j - m} {l \choose n} {N-l \choose k - n} \text{ (since $N$ is odd)}. 
\end{align*}
This means that
\begin{align*}
&\sum_{k=0}^{(N-1)/2} \sum_{m=0}^{\min\{j, k\}} \sum_{n=0}^{\min\{k,l\}} (-1)^{k(j+l)+m+n} {k \choose m} {N - k \choose j - m} {l \choose n} {N-l \choose k - n} \\
&= \frac{1}{2} \sum_{k=0}^N \sum_{m=0}^{\min\{j, k\}} \sum_{n=0}^{\min\{k,l\}} (-1)^{k(j+l)+m+n} {k \choose m} {N - k \choose j - m} {l \choose n} {N-l \choose k - n} \\
&= \begin{cases}
2^{N-1} &\text{ if $j = l$} \\
0 &\text{ otherwise},
\end{cases}
\end{align*}
which implies that $M^2 = 1$ in corner algebras (recall that $M$ has the same definition as in the subordinate case except that $2^{N/2}$ is replaced by $2^{(N-1)/2}$ in (\ref{defM})).

\section{Some Cross-Symmetry Matrices} \label{csex}
\subsection{For type $\R$ corner algebras}
\begin{longtable}{llll}
$p+q = 3:$ & $ \frac{1}{2} \left( \begin {array}{cc} 1&1\\\noalign{\medskip}3&-1\end {array} \right)$ &
$p+q = 5:$ & $ \frac{1}{4} \left( \begin {array}{ccc} 1&1&-1\\\noalign{\medskip}5&-3&-1
\\\noalign{\medskip}-10&-2&-2\end {array}
 \right) $ \\ \\

$p+q = 7:$ & $  \frac{1}{8} \left( \begin {array}{cccc} 1&1&-1&-1\\\noalign{\medskip}7&-5&-3&1
\\\noalign{\medskip}-21&-9&1&-3\\\noalign{\medskip}-35&5&-5&3
\end {array}
 \right) $ & $p+q = 9:$ & $ \frac{1}{16} \left( \begin {array}{ccccc} 1&1&-1&-1&1\\\noalign{\medskip}9&-7&-5&3
&1\\\noalign{\medskip}-36&-20&8&0&4\\\noalign{\medskip}-84&28&0&8&4
\\\noalign{\medskip}126&14&14&6&6\end {array}
 \right) $ \\ \\
 \end{longtable}
 
\begin{longtable}{cc} 
$p+q = 11:$ & $ \frac{1}{32} \left( \begin {array}{cccccc} 1&1&-1&-1&1&1\\\noalign{\medskip}11&-9&
-7&5&3&-1\\\noalign{\medskip}-55&-35&19&7&1&5\\\noalign{\medskip}-165&
75&21&5&11&-5\\\noalign{\medskip}330&90&6&22&-6&10\\\noalign{\medskip}
462&-42&42&-14&14&-10\end {array}
 \right) $ \\ \\
 
 $p+q = 13:$ & $ \frac{1}{64} \left( \begin {array}{ccccccc} 1&1&-1&-1&1&1&-1\\\noalign{\medskip}13
&-11&-9&7&5&-3&-1\\\noalign{\medskip}-78&-54&34&18&-6&2&-6
\\\noalign{\medskip}-286&154&66&-14&10&-14&-6\\\noalign{\medskip}715&
275&-55&25&-29&-5&-15\\\noalign{\medskip}1287&-297&33&-63&-9&-25&-15
\\\noalign{\medskip}-1716&-132&-132&-36&-36&-20&-20\end {array}
 \right) $ \\ \\

$p+q = 15:$ & $ \frac{1}{128} \left( \begin {array}{cccccccc} 1&1&-1&-1&1&1&-1&-1
\\\noalign{\medskip}15&-13&-11&9&7&-5&-3&1\\\noalign{\medskip}-105&-77
&53&33&-17&-5&-3&-7\\\noalign{\medskip}-455&273&143&-57&-7&-15&-17&7
\\\noalign{\medskip}1365&637&-221&-21&-43&-35&3&-21
\\\noalign{\medskip}3003&-1001&-143&-99&-77&-1&-39&21
\\\noalign{\medskip}-5005&-1001&-143&-187&11&-65&25&-35
\\\noalign{\medskip}-6435&429&-429&99&-99&45&-45&35\end {array}
 \right) $  \\ \\

$p+q = 17:$ & $ \frac{1}{256} \left( \begin {array}{ccccccccc} 1&1&-1&-1&1&1&-1&-1&1
\\\noalign{\medskip}17&-15&-13&11&9&-7&-5&3&1\\\noalign{\medskip}-136&
-104&76&52&-32&-16&4&-4&8\\\noalign{\medskip}-680&440&260&-132&-48&0&-
20&20&8\\\noalign{\medskip}2380&1260&-560&-168&-12&-60&40&0&28
\\\noalign{\medskip}6188&-2548&-728&0&-156&84&-16&56&28
\\\noalign{\medskip}-12376&-3640&364&-364&208&-32&100&28&56
\\\noalign{\medskip}-19448&3432&-572&572&0&176&44&84&56
\\\noalign{\medskip}24310&1430&1430&286&286&110&110&70&70\end {array}

 \right) $
 \end{longtable}
 
 \subsection{For type $\R$ subordinate algebras}
 \begin{longtable}{llrr}
 $p+q = 2:$ & $\frac{1}{2} \left( \begin {array}{ccc} 1&1&-1\\\noalign{\medskip}2&0&2
\\\noalign{\medskip}-1&1&1\end {array} \right)$ & 
$p+q = 4:$ & $\frac{1}{4}  \left( \begin {array}{ccccc} 1&1&-1&-1&1\\\noalign{\medskip}4&-2&0&-2
&-4\\\noalign{\medskip}-6&0&-2&0&-6\\\noalign{\medskip}-4&-2&0&-2&4
\\\noalign{\medskip}1&-1&-1&1&1\end {array} \right)$ \\ \\
\end{longtable}

\begin{longtable}{cc}
$p+q = 6:$ & $\frac{1}{8}  \left( \begin {array}{ccccccc} 1&1&-1&-1&1&1&-1\\\noalign{\medskip}6&
-4&-2&0&-2&4&6\\\noalign{\medskip}-15&-5&-1&-3&1&-5&15
\\\noalign{\medskip}-20&0&-4&0&-4&0&-20\\\noalign{\medskip}15&-5&1&-3&
-1&-5&-15\\\noalign{\medskip}6&4&-2&0&-2&-4&6\\\noalign{\medskip}-1&1&
1&-1&-1&1&1\end {array} \right)$ \\ \\

$p+q = 8:$ & $\frac{1}{16} \left( \begin {array}{ccccccccc} 1&1&-1&-1&1&1&-1&-1&1
\\\noalign{\medskip}8&-6&-4&2&0&2&4&-6&-8\\\noalign{\medskip}-28&-14&4
&-2&4&2&4&14&-28\\\noalign{\medskip}-56&14&-4&6&0&6&4&14&56
\\\noalign{\medskip}70&0&10&0&6&0&10&0&70\\\noalign{\medskip}56&14&4&6
&0&6&-4&14&-56\\\noalign{\medskip}-28&14&4&2&4&-2&4&-14&-28
\\\noalign{\medskip}-8&-6&4&2&0&2&-4&-6&8\\\noalign{\medskip}1&-1&-1&1
&1&-1&-1&1&1\end {array} \right)$

 \end{longtable} 

\bibliographystyle{utphys}
\bibliography{refs}

\end{document}